\documentclass[10pt,pdfa]{article}
\usepackage[in]{fullpage}

\usepackage[style=alphabetic,minalphanames=3,maxalphanames=4,maxnames=99,backref=true]{biblatex}

\usepackage{url}
\usepackage{bbm}
\usepackage{hyperref} 
\hypersetup{breaklinks=true}

\usepackage[normalem]{ulem}

\def\O{\mathcal{O}}

\usepackage{mdframed}
\usepackage{dsfont}
\usepackage{tabularx}
\usepackage[normalem]{ulem}
\usepackage{comment}
\usepackage[shortlabels]{enumitem}
\usepackage{titlesec}
\usepackage{bm}
\usepackage{amsfonts}
\usepackage{xspace}
\usepackage{amsmath}
\usepackage{amssymb}
\usepackage{amsthm}
\usepackage{xcolor}
\usepackage{graphicx}
\usepackage{mathrsfs}
\usepackage{qtree}
\usepackage{tree-dvips}
\usepackage{float}
\usepackage{physics}
\usepackage{authblk}
\usepackage{braket}
\usepackage[T1]{fontenc}
\usepackage{mathtools}
\usepackage{mathrsfs}
\usepackage{tikz}
\usepackage[linesnumbered,ruled,vlined]{algorithm2e}
\usepackage{qcircuit}
\usepackage{algorithmic}
\usepackage[nameinlink,capitalize]{cleveref}

\newcommand{\subversion}[1]{}
\usepackage{bm} 

  \hypersetup{colorlinks={true},linkcolor={blue},citecolor=magenta}

  \theoremstyle{plain}
  \newtheorem{theorem}{Theorem}[section]
  \newtheorem{lemma}[theorem]{Lemma}
  \newtheorem{corollary}[theorem]{Corollary}
  \newtheorem{definition}[theorem]{Definition}
  \newtheorem{remark}[theorem]{Remark}
  \newtheorem{claim}[theorem]{Claim}
  
 \newtheorem{construction}{Construction}

\usepackage[style=alphabetic,minalphanames=3,maxalphanames=4,maxnames=99,backref=true]{biblatex}

\newtheorem*{rep@theorem}{\rep@title}
\newcommand{\newreptheorem}[2]{%
\newenvironment{rep#1}[1]{%
 \def\rep@title{#2 \ref{##1}}%
 \begin{rep@theorem}}%
 {\end{rep@theorem}}}
\makeatother

\newreptheorem{theorem}{Theorem}
\newreptheorem{lemma}{Lemma}
\newreptheorem{corollary}{Corollary}

\newcommand{\bit}{\{0,1\}}

\newcommand{\ignore}[1]{}

\newcommand{\proj}[1]{\ensuremath{|#1\rangle \langle #1|}}

\newcommand{\E}{\mathop{\mathbb{E}}}

\newcommand{\N}{\mathbb{N}}

\newcommand{\negl}{\mathsf{negl}}

\renewcommand{\cal}[1]{\mathcal{#1}}

\newcommand\id{\mathbb{I}}

\newcommand{\TD}{\mathsf{TD}}

\newcommand{\poly}{\mathrm{poly}}
\newcommand{\eps}{\epsilon}

\DeclarePairedDelimiterX{\hsip}[2]{\langle}{\rangle_{\tiny{\mathtt{HS}}}\xspace}{#1, #2}

\newcommand{\KeyGen}{\mathsf{KeyGen}}
\newcommand{\Enc}{\mathsf{Enc}}
\newcommand{\Dec}{\mathsf{Dec}}

\newcommand\algo{\mathcal}

\newif\ifsubmission
\submissionfalse
\ifsubmission

\else

\newcommand{\alex}[1]{{\noindent \textcolor{orange}{\emph{(Alex:  #1)}}}{}}

\fi

\newcommand{\setup}{\mathsf{Setup}}
\newcommand{\enc}{\mathsf{Enc}}
\newcommand{\dec}{\mathsf{Dec}}
\newcommand{\ct}{\mathsf{ct}}
\newcommand{\vk}{\mathsf{vk}}
 
\newcommand{\adversary}{{\cal A}}
\newcommand{\sk}{\mathsf{sk}}

\newcommand{\prob}{\mathsf{Pr}}
\newcommand{\secparam}{\lambda}
\newcommand{\revokeexperiment}{\mathsf{RevokeExpt}}
\newcommand{\oracle}{\mathcal{O}}
\newcommand{\perm}{{\bm \sigma}}
\newcommand{\hybrid}{\mathbf{H}}
\newcommand{\Compile}{\mathsf{Compile}}
\newcommand{\Eval}{\mathsf{Eval}}
\newcommand{\Revoke}{\mathsf{Revoke}}

\title{Revocable Encryption, Programs, and More:\\ The Case of Multi-Copy Security}
\author[1]{Prabhanjan Ananth}
\author[2]{Saachi Mutreja}
\author[3]{Alexander Poremba}
\affil[1]{University of California, Santa Barbara}
\affil[2]{Columbia University}
\affil[3]{Massachusetts Institute of Technology}
\date{} 

\begin{document}

\maketitle

\begin{abstract}
\noindent Fundamental principles of quantum mechanics have inspired many new research directions, particularly in quantum cryptography. One such principle is \emph{quantum no-cloning} which has led to the emerging field of revocable cryptography. Roughly speaking, in a revocable cryptographic primitive, a cryptographic object (such as a ciphertext or program) is represented as a quantum state in such a way that surrendering it effectively translates into losing the capability to use this cryptographic object. All of the revocable cryptographic systems studied so far have a major drawback: the recipient only receives one copy of the quantum state. Worse yet, the schemes become completely insecure if the recipient receives many identical copies of the {\em same} quantum state---a property that is clearly much more desirable in practice. 
\par While multi-copy security has been extensively studied for a number of other quantum cryptographic primitives, it has so far received only little treatment in context of unclonable primitives. Our work, for the first time, shows the feasibility of revocable primitives, such as revocable encryption and revocable programs, which satisfy multi-copy security in oracle models. This suggest that the stronger notion of multi-copy security is within reach in unclonable cryptography more generally, and therefore could lead to a new research direction in the field.

\end{abstract}

\newpage
\section{Introduction}
Designing mechanisms to provably revoke cryptographic capabilities is an age-old problem~\cite{stubblebine1995recent,rivest1998can}. In the public-key infrastructure, certificate authorities have the ability to invalidate public-key certificates~\cite{certrevoke24}, especially when the certificates have been compromised. Key rotation policies~\cite{GC24} guarantee that outdated decryption keys become ineffective for future use. The existing approaches to tackle with this problem have their limitations owing to the fact that cryptographic secrets are represented as binary strings and hence, it is infeasible to provably ensure that the malicious attackers have erased information from their devices. In the context of key rotation, a compromised key can still be used to decrypt old ciphertexts. Another issue with using classical information to represent cryptographic keys is that it is difficult to detect compromise: a hacker could steal the classical key from a device without leaving a trace. 
\par Recently, a line of works~\cite{cryptoeprint:2013/606,ALP21,AKNYY23,APV23,CGJL23,MPY23,AHH24} have leveraged quantum information and proposed new approaches for provable revocation of cryptographic objects, such as ciphertexts, programs and keys. These works studied revocation in the context of many Crypto 101 primitives, including pseudorandom functions, private-key and public-key encryption and digital signatures. In a revocable cryptographic primitive, an object (such as a program or a decryption key, etc.) is associated with a quantum state in such a way that \emph{only with access to the state} the functionality of the original cryptographic object is retained. The common template for defining revocable security is in the form of a cryptographic game: The adversary receives \emph{one copy} of the quantum state that it can use for a limited period of time after which it is supposed to return back the state to the owner. The security guarantee stipulates that after the state is returned, the adversary effectively loses the capability to use the cryptographic object. At this point, it should be clear to the reader the necessity that such cryptographic objects are represented as quantum states: indeed, if they were classical, the adversary could always maintain a secret copy, while pretending to have erased everything from its device.  On the other hand, the no-cloning theorem~\cite{WZ82,Dieks82} of quantum mechanics suggests that the above security experiment could very well be achieved. 

\paragraph{Multi-Copy Security.} Let us now zoom in on the part of the security experiment, where the adversary receives {\em only one copy} of the quantum state. In all of the prior works in the literature so far~\cite{AKNYY23,APV23,CGJL23,MPY23,AHH24}, this limitation persists. One could consider a more general definition, where the adversary receives $k$ {\em identical} copies of the quantum state and is later asked to return back all of the copies of the state.
The security guarantee is similar to before: after returning all of the $k$ copies, the adversary should effectively lose all access to the underlying crytptographic object. We term this general security experiment to be {\em multi-copy security}. 
\par There are a couple of reasons to study multi-copy security for revocable primitives. 
\begin{itemize}
    \item \underline{\textsc{Historical Context}}:  Multi-copy security is not new and has been extensively studied in quantum cryptography, especially in the context of foundational primitives such as pseudorandom states~\cite{JLS18,BS19,AQY22,AGQY22,bostanci2023unitarycomplexityuhlmanntransformation,bostanci2024efficientquantumpseudorandomnesshamiltonian} and one-way state generators~\cite{MY22,MY23}. Indeed, multi-copy security has been crucial in the design of many cryptographic constructions. The works of~\cite{AGQY22,ALY24,KT24} used tomography, which inherently requires multiple copies in order to dequantize the communication in some of the quantum cryptographic primitives. Specifically for revocable primitives, a conceptual reason to study multi-copy security is to understand whether having more copies necessarily makes it easier for the adversary to clone quantum states.  Investigating multi-copy security for revocable primitives is a starting step towards understanding multi-copy security for more advanced primitives such as public-key quantum money~\cite{AC12,Zha21}. The question of whether multi-copy security is possible in unclonable cryptography was also recently raised in~\cite{metger2024simpleconstructionslineardepthtdesigns}.
    
    \item \underline{\textsc{Nested Leasing}}: Having access to many more copies of the quantum state would also give more power to the user; for example, using a permutation test~\cite{KNY08} (a generalized version of SWAP test) where one is given $\ket{\phi}$ and polynomially many copies of $\ket{\psi}$, one can approximately test the overlap between $\ket{\phi}$ and $\ket{\psi}$. This ability allows for nested leasing of cryptographic objects, such as programs or keys. Suppose a user \textbf{A} is leased a large number of, say $k$, copies of a quantum program $\ket{\psi}$. User \textbf{A} could further lease a number, say $k' \ll k$, of copies of $\ket{\psi}$ to user \textbf{B}. At a later point in time, when user \textbf{B} is asked to return back its copies to user \textbf{A}. User \textbf{A} can then use its $k-k'$ copies to approximately test whether the returned copies are correct. If the test succeeds, user \textbf{A} then is in a position to return back all of its $k$ copies to the true owner\footnote{A drawback of this approach is that there is some room for user \textbf{B}  to cheat with noticeable probability in this approach without user \textbf{A} noticing, which means that the owner would not always be to pinpoint whether user \textbf{A} or user \textbf{B}  cheated. Still, this approach offers a non-trivial solution to this challenging problem.}. Such a nested leasing approach could be especially useful in organizations with hierarchical structure. 
\end{itemize}
\noindent The notion of multi-copy security we consider in this work is closely related to {\em collusion-resistant security}, considered in the works of~\cite{LLLZ22,CG23}. The crucial difference is that in the prior works, the adversary receives {\em i.i.d} copies of the quantum key whereas in our case, the adversary receives {\em identical} copies of the quantum key. In the nested leasing application discussed above, it was crucial that the user received many {\em identical} copies of the same state.

\paragraph{Multi-copy security using commonly studied unclonable states: Challenges.} The first step towards addressing multi-copy security is to identify quantum states that are unclonable. We discuss the commonly studied unclonable quantum states below: 
\begin{itemize}
    \item \underline{\textsc{BB84 states}}: these states are of the form $H^{\theta}\ket{x}$, where $\theta \in \{0,1\}^n$ and $x \in \{0,1\}^n$. These states have been influential in the design of private-key quantum money~\cite{Wiesner83} and in the design of encryption schemes with unclonable ciphertexts~\cite{BL19,BI20}. Given many copies of $H^{\theta}\ket{x}$, one can learn $x$ and $\theta$ and hence, recover a complete description of $H^{\theta}\ket{x}$. 
    \item \underline{\textsc{Subspace and Coset states}}: these states are of the form $(\sqrt{|A|})^{-1} \cdot \sum_{{\bf x} \in A} \ket{{\bf x}}$, where $A \subseteq \mathbb{F}_2^n$ is a sparse subspace of $\mathbb{F}_2^n$, for some $n \in \mathbb{N}$. These states have been crucial in the design of public-key quantum money~\cite{AC12,Zha21}, among other primitives. Again given many copies, one can learn the basis of the subspace and hence a complete description of the subspace state. Another related class of unclonable states are coset states which are superpositions over a coset (rather than a subspace) and moreover, each term in the superposition has a phase that depends on a dual coset. Coset states have been influential in constructions of quantum copy-protection~\cite{Coladangelo_hidden}. Similar to subspace states, coset states are also learnable. 
    \item \underline{\textsc{SIS-based states}}: these states are of the form $\sum_{{\bf x} \in \mathbb{Z}^m_q,{\bf A}{\bf x}={\bf y}} \alpha_{{\bf x}} \ket{{\bf x}}$, where $q,m \in \mathbb{N}$, $|\alpha_{{\bf x}}|^2$ is a discrete Gaussian distribution such that most of the weight is on low norm vectors ${\bf x}$ and finally, ${\bf A} \in \mathbb{Z}_q^{n \times m},{\bf y} \in \mathbb{Z}_q^n$. They were useful in designing traditional and advanced encryption systems with unclonable quantum keys~\cite{Por23,APV23,MPY23,AHH24}. Given many copies of this state, one can recover a short basis of the kernel of ${\bf A}$ which can then be used to recover the above state. 
\end{itemize}
\noindent In other words, all the above types of states are learnable and hence, they cannot be the basis of any unclonable cryptographic scheme satisfying multi-copy security. This suggests that we need to look for new unclonable quantum states that are unlearnable even given many copies. In the past, discovering new unclonable quantum states has led to pushing the frontier of unclonable quantum cryptographic primitives and we believe our endeavour could reap similar results.

\section{Our Results}

We now give an overview of our results.

\paragraph{Our Approach: Quantum Pseudorandomness Meets Unclonable Cryptography.} We use {\em subset states} to tackle multi-copy security. Subset states have been recently studied in the context of quantum pseudorandomness~\cite{GB23,JMW23,aaronson_et_al:LIPIcs.ITCS.2024.2}. A subset state is associated with an unstructured and random subset $S \subset \{0,1\}^n$ of the form $\ket{S} = (\sqrt{|S|})^{-1} \sum_{x \in S} \ket{x}$. Several recent works~\cite{GB23,JMW23} showed that random subset states (of non-trivial size) are approximate state $k$-designs for any polynomial $k$ as a function of $n \in \mathbb{N}$, which make random subset states\footnote{Strictly speaking, we use \emph{pseudorandom} subset states which can be generated efficiently via pseudorandom permutations.} a natural candidate for multi-copy security---particularly since Haar states are unlearnable given polynomially many identical copies. At first sight, it would seem as though that the fact that random subset states are close to Haar states should be discourage us from using them for unclonable cryptography, especially since Haar states have virtually no structure. Indeed, the structure of BB84 states, subspace states and others has been crucially exploited in various applications. Our work shows that subset states, in some sense, have the minimal amount of structure to enable 
 a number of interesting cryptographic primitives in the context of multi-copy revocable cryptography. To the best of our knowledge, these applications also mark the first use case of subset states in the context of cryptography. Our main technical contribution is of \emph{information-theoretic} nature: we prove a query lower bound for \emph{forging} subset elements; concretely, we show that any quantum algorithm
that receives $k$ copies of a random subset state $\ket{S}$ cannot produce $k+1$ many subset elements in $S$ unless it makes a large amount of queries to a membership oracle for $S$. We believe that this result could be of independent interest.

\paragraph{Multi-Copy Revocable Encryption.} We first study revocable encryption~\cite{cryptoeprint:2013/606,AKNYY23,APV23} with multi-copy security. A revocable encryption scheme is a regular encryption scheme but where the ciphertexts are associated with quantum states. Additionally, a revocable encryption scheme comes with the following security notion called {\em multi-copy revocable security}: informally, it states that any adversary that successfully returns $k$ valid copies of a quantum ciphertext which it was given by a challenger, where $k$ is an arbitrary polynomial, necessarily loses the ability to decrypt the ciphertexts in the future---even if the secret key is revealed. In more detail, the security game is formulated as follows: 
\begin{itemize}
    \item The adversary selects a pair of messages $(m_0,m_1)$ and sends them to the challenger.

    \item The challenger randomly selects one of the two messages, say $m_b$, and encrypts it $k$ times using a secret ket $\mathsf{sk}$, and sends the ciphertext copies $\ket{\psi_b}^{\otimes k}$ to the adversary $\adversary$. 
    \item At a later point in time, $\adversary$ returns back all the copies of $\ket{\psi_b}$, which are then verified by the challenger.
    \item After successfully returning back the states, $\adversary$ receives the secret key $\mathsf{sk}$ in the clear. Finally, $\adversary$ outputs a guess $b'$. 
\end{itemize}
The scheme is said to be secure if the probability that $b'=b$ is close to $1/2$. 
\par Prior works~\cite{cryptoeprint:2013/606,AKNYY23,APV23,CGJL23,AHH24} only studied variants of the above security game in the setting where the adversary receives only one copy of the quantum ciphertext. In fact, their schemes are easily seen to not satisfy multi-copy security. We show the following. 

\begin{theorem}
\label{thm:mcs:rp}
If post-quantum one-way functions exist, then there exists an encryption scheme with (an oracular notion of) multi-copy revocable security.
\end{theorem}

\noindent Some remarks are in order. Firstly, our security proof does not fully achieve the standard notion of revocable security guarantee we stated above; rather, we consider a slightly different variant of the experiment, where in the second part of the game (instead of revealing the secret key in the clear) we allow the adversary to query an oracle that is powerful enough to enable decryption during the first phase of the game. 
While this constitutes a weaker notion of security, it nevertheless results in a meaningful notion of revocable security: once the adversary has successfully returned all of the copies of the ciphertext, it can no longer decrypt the ciphertext in the future---even if it gets access to an oracle that would have previously allowed it to do so.
Second, our construction of multi-copy revocable encryption makes use of quantum-secure pseudorandom permutations (QPRPs), which can be constructed from post-quantum one-way functions~\cite{zhandry2016notequantumsecureprps}. In the security experiment which underlies our construction, the aforementioned oracle for decryption (i.e., that which is handed to the adversary after revocation has taken place) is in the form of an ideal oracle for the permutation itself. Once again, the rationale behind our notion of oracular security is that an attacker who receives a QPRP key in the clear would most certainly use it to evaluate the QPRP, and hence it is reasonable to consider a model in which the attacker receives an oracle for the permutation instead.\footnote{Note, however, that this does not capture all possible attacks; for example, the adversary could use its knowledge of the QPRP key to break the scheme in other meaningful ways. 
} A key advantage of oracular security is that we can directly invoke the security of the QPRP and use a perfectly random permutation instead. \footnote{This switch is generally not possible in the standard notion of revocable security in which the QPRP key is required to revealed in the clear. Here, QPRP security does not apply.}
We remark that this model loosely resembles the random permutation model behind the international hash function standard SHA-3~\cite{KeccakSponge3,KeccakSub3}, except that the adversary only receives oracle access to the permutation during the second part of the revocable security experiment.
While our construction only achieves an oracular notion of revocable security, it is nevertheless the very first construction of revocable encryption which satisfies multi-copy security {\em in any model}.

\paragraph{Multi-Copy Revocable Programs.} Our previous discussion on revocable encryption illustrates that encryption and decryption functionalities can be protected even if many copies of the quantum ciphertext are made available to the recipient. We generalize this result further and study whether arbitrary functionalities can be protected. We define and study revocable programs with multi-copy security. In this notion, there is a functionality preserving compiler that takes a program and converts it into a quantum state. The security guarantee is defined similar to revocable encryption:  
\begin{itemize}
    \item The challenger compiles a program $P$, sampled from a distribution ${\cal D}$ on a set of programs ${\cal P}$, into a state $\ket{\psi_P}$. It then  sends $k$ copies of the state $\ket{\psi_P}$ to the adversary $\adversary$. 
    \item At a later point in time, $\adversary$ returns back all of the copies of $\ket{\psi_P}$. 
    \item After returning back the state, $\adversary$ is given $x$, where $x$ is sampled from the input distribution of $P$. It then outputs a guess $y$. 
\end{itemize}
\par The scheme is said to be secure if the probability that $y=P(x)$ is roughly close to the trivial success probability. Here, the trivial success probability is defined as the optimal probability of guessing $P(x)$ given just $x$ (and the knowledge of ${\cal D}$ and the input distribution).
\par Prior works propose revocable programs for specific functionalities in the plain model~\cite{Coladangelo_hidden} or for general functionalities in oracle models~\cite{ALLZZ21}. However, these works guarantee security only if the adversary receives one copy of the state. And as before, these constructions provably do not satisfy multi-copy security. We show the following. 

\begin{theorem}
There exist revocable programs which satisfy (an oracular notion of) multi-copy security in a classical oracle model. 
\end{theorem}

\noindent Unlike~\Cref{thm:mcs:rp}, the above theorem relies upon structured and ideal classical oracles. 
\par Finally, for the special case of point functions, we show that we can again only rely upon pseudorandom permutations together with the (standard) quantum random oracle model. We show the following.  

\begin{theorem}
There exist revocable multi-bit point functions which satisfy (an oracular notion of) multi-copy security in quantum random oracle model. 
\end{theorem}

\noindent Our results and techniques opens the door for building more advanced unclonable primitives that preserve their security even if the adversary receives many copies of the unclonable quantum state. 

\paragraph{Are Results in Oracle Models Interesting?} It is natural for a reader to be skeptical of our results given that they are based in the oracle models. However, we would like to emphasize that achieving results in the oracle models still requires non-trivial amount of effort. As history suggests, constructions in the oracle models have eventually been adopted to constructions in the plain model. A classic example is the construction of public-key quantum money, which was first proposed in the oracle models by Aaronson and Christiano~\cite{AC12} and later, being instantiated in the plain model by Zhandry~\cite{Zha21}. In a similar vein, our techniques could be useful for future works on achieving multi-copy security in the plain model.


\paragraph{Applications to Sponge Hashing.}

As a complementary contribution, we show that the techniques we developed in this paper are more broadly applicable and extend to other cryptographic settings as well. Here, we single out the so-called \emph{sponge construction} used in SHA-3~\cite{KeccakSponge3,KeccakSub3}.

In \Cref{sec:sponge}, we study a simple query problem: Suppose that an adversary receives
as input a hash table for a set of random input keys, where each hash is computed using a \emph{salted} (one-round) sponge hash function. How many quantum queries are necessary to find a new \emph{valid} element in the range of the hash function?
Our contribution is a space-time trade-off which precisely characterizes the hardness of finding hash table elements in the presence of oracles that depend non-trivially on the sponge hash function.

\paragraph{Acknowledgements.}

The authors would like to thank Henry Yuen and Tal Malkin for many useful discussions. PA
is supported by the National Science Foundation under Grant No. 2329938 and Grant No.
2341004. SM is supported by
AFOSR award FA9550-21-1-0040, NSF CAREER award CCF-2144219, and the Sloan Foundation.
AP is supported by the U.S. Department of Energy, Office of Science, National Quantum Information Science Research Centers, Co-design Center for Quantum Advantage (C2QA) under contract number DE-SC0012704.

\subsection{Related work}

We now discuss related notions which are relevant to this work.

\paragraph{Copy-Protection.} This notion was first introduced by Aaronson~\cite{Aar09}. Informally speaking, a copy-protection scheme is a
compiler that transforms programs into quantum states in such a way that using the resulting states, one can run the original program. Yet, the security guarantee stipulates that any adversary given one copy of the state cannot produce a bipartite state wherein both parts compute the original program. Copy-protection schemes have since been constructed for various classes of programs and under various different models, for example as in~\cite{Aar09,ALLZZ21,Coladangelo2024quantumcopy,AK22,AKL23,LLLZ22,CG23}. We remark, however, that all of the aforementioned works are completely insecure if multiple identical copies of the program are made available.
The notion of multi-copy security we consider in this work is closely related to {\em collusion-resistant security}, considered in the works of~\cite{LLLZ22,CG23}. The crucial difference is that in the prior works, the adversary receives {\em i.i.d} copies of the quantum key whereas in our case, the adversary receives {\em identical} copies of the quantum key.

\paragraph{Secure Software Leasing.} Another primitive relevant to revocable cryptography is secure software leasing~\cite{ALP21}. The notion of secure software leasing states that any program can be compiled into a functionally equivalent program, represented as a quantum state, in such a way that once the compiled program is returned, the (honest) evaluation algorithm on the residual state cannot compute the original functionality. Secure leasing has been constructed for various functionalities~\cite{ALP21,Coladangelo2024quantumcopy,BJLPS21,KNY21}.
Similar to copy-protection, none of the aforementioned works consider multi-copy security.

\paragraph{Encryption Schemes with Revocable Ciphertexts.} 
Unruh \cite{cryptoeprint:2013/606} proposed a (private-key) quantum timed-release
encryption scheme that is \emph{revocable}, i.e. it allows a user to \emph{return} the ciphertext of a quantum timed-release encryption scheme, thereby losing all access to the data. 
Broadbent and Islam~\cite{BI20} introduced the notion of \emph{certified deletion}, which is incomparable with the related notion of unclonable encryption. This has led to the development of other certified deletion protocols, for example as in Ref.~\cite{Por23,10.1007/978-3-031-38554-4_7,10.1007/978-3-031-38554-4_4,cryptoeprint:2023/559,BBSS23}.
However, the notion of multi-copy security, such as in our work, has not been studied.

\section{Technical Overview}
\noindent We first discuss the main technical lemma that underlies the construction of all the revocable primitives. 

\paragraph{$k \rightarrow k+1$ Unforgeability Of Subset States.} Given $k$ copies of a random subset state $\ket{S}$, where $|S| = \omega(\log(n))$, it is easy to produce $k$ distinct elements in the set $S \subseteq \{0,1\}^n$. This can be seen by just measuring all the copies of $\ket{S}$ in the computational basis. Since $|S|$ is sufficiently large, the probability that we get all the $k$ elements to be distinct is high.  Our main lemma states that it is computationally infeasible to produce $k+1$ distinct elements in $S$ given $k$ copies of $\ket{S}$ and access to an oracle that tests membership in $S$. Roughly speaking, this can be seen as a multi-copy subset state analogue of the complexity-theoretic no-cloning result for \emph{subspace states}, which was shown by Aaronson and Christiano~\cite{AC12}.

To prove this, we appeal to the fact that random subset states are indistinguishable from Haar states, as shown in a couple of recent works~\cite{JMW23,GB23}. That is, for any $k$ polynomial in $n$ $$\underset{\substack{S \subseteq \{0,1\}^n\\ |S|=\omega(\log(n))}}{\mathbb{E}}\left[ \ketbra{S}{S}^{\otimes k} \right] \approx \underset{\ket{\psi} \leftarrow {\cal H}_n}{\mathbb{E}}\left[ \ketbra{\psi}{\psi}^{\otimes k} \right]$$
Here, $\approx$ refers to (negligible in $n$) closeness in trace distance and ${\cal H}_n$ refers to the Haar distribution on $n$-qubit quantum states. 
\par Furthermore, using the characterization of the symmetric subspace~\cite{Harrow13}, the following holds:  
$$\underset{\substack{S \subseteq \{0,1\}^n\\ |S|=\omega(\log(n))}}{\mathbb{E}}\left[ \ketbra{S}{S}^{\otimes k} \right] \approx \underset{\substack{{\bf x}=(x_1,\ldots,x_k) \leftarrow \{0,1\}^{nk}\\ x_1,\ldots,x_k\text{ are distinct}}}{\mathbb{E}}\left[ \ketbra{\sigma_{\bf x}}{\sigma_{{\bf x}}} \right],$$
Here, $\ket{\sigma_{{\bf x}}} = \frac{1}{k!}\sum_{\sigma \in S_k} \ket{x_{\sigma(1)},\ldots,x_{\sigma(k)}}$, where $S_k$ is a symmetric group on $k$ elements. 
\par Why does this characterization help? note that if one were given $k$ copies of $\ket{\sigma_{{\bf x}}}$, it is impossible to produce $k+1$ distinct elements in $\{x_1,\ldots,x_k\}$ because there are only $k$ of them! This observation was first made in~\cite{BBSS23} although they considered subset states {\em with phase} for applications that are unrelated to our work. 
\par However this observation is alone not sufficient to complete the proof for two main reasons. Firstly, the adversary receives access to the subset membership oracle. Secondly, the winning conditions are different in both the cases; in one case, given $k$ copies of $\ket{S}$, $k+1$ elements in $S$ had to be produced while in the other case, given $\ket{\sigma_{{\bf x}}}$, $k+1$ elements in ${\bf x}$ need to be produced. \\

\noindent \textsc{Intermediate Lemma.} To make progress on this lemma, we first study the following intermediate lemma. Consider the following two distributions:
\begin{itemize}
    \item ${\cal D}_0$: Sample $t$ copies of a random subset state $\ket{S}$, i.e. $S$ is sampled uniformly at random from all sets of size $\omega(\log(n))$. 
    \item ${\cal D}_1$: Sample $\ket{\sigma_{{\bf x}}}$, where ${\bf x}=(x_1,\ldots,x_k) \leftarrow \{0,1\}^{nk}$. 
\end{itemize}
Similarly, we can define two oracles. The first oracle is ${\cal O}_0$ and the second oracle is ${\cal O}_1$. We define ${\cal O}_0$ to test membership in $S$ and ${\cal O}_1$ to test membership in $\{x_1,\ldots,x_k\}$. The intermediate lemma states that a query-bounded adversary cannot distinguish whether it receives a sample from ${\cal D}_0$ and oracle access to ${\cal O}_0$ or it receives a sample from ${\cal D}_1$ and oracle access to ${\cal O}_1$. 
\par We prove this lemma in a series of steps. First, the adversary receives a sample from ${\cal D}_0$ and has oracle access to ${\cal O}_0$. 
\begin{itemize}
    \item We then modify the oracle that the adversary has access to. Instead of having oracle access to $S$, it instead has oracle access to $T$, where $S \subseteq T$ and moreover, $S$ has negligible size compared to $T$ and $T$ has negligible size compared to $2^n$. The intuition is that the adversary cannot find an element in $T \backslash S$ and then, thanks to the O2H lemma~\cite{unruh2007random}, the indistinguishability of the membership oracles (for $S$ and $T$) follows. 
    \item We change the process of sampling $S$. Instead of first sampling $S$ and then sampling $T$ subject to $S \subseteq T$. We first sample $T$ and then sample $S$ to be a random subset of $T$ of size $\omega(\log(n))$. This process is identical to the previous step. 
    \item Then, we switch $k$ copies of $\ket{S}$ with $\ket{\sigma_{{\bf x}}}$, for a random ${{\bf x}}=(x_1,\ldots,x_k)$, where $(x_1,\ldots,x_k)$ are drawn from $T$. Here, we prove a stronger statement that {\em even if the adversary has a description of $T$}, this indistinguishability holds. In other words, the adversary receives a sample from ${\cal D}_1$ although it has access to an oracle that tests membership in $T$. This follows from some propositions in~\cite{JMW23}. 
    \item Finally, we modify the previous step as follows. The adversary receives a sample from ${\cal D}_1$ with oracle access to ${\cal O}_1$. To prove that the adversary cannot notice the modification, we need to show that the adversary cannot distinguish whether it has oracle access to membership in $T$ or membership in $\{x_1,\ldots,x_k\}$, where $x_1,\ldots,x_k \leftarrow T$. In fact, we need to show that the indistinguishability should hold {\em even if the adversary receives $x_1,\ldots,x_k$ in the clear}. This again follows from another invocation of the O2H lemma. 
\end{itemize}

\noindent The above intermediate lemma immediately implies the $k \rightarrow k+1$ unforgeability of subset states. Let us see why. Suppose there is an adversary $\adversary$ that violates the $k \rightarrow k+1$ unforgeability property with non-negligible probability $p$. Then we can come up with a reduction ${\cal R}$ that contradicts the above intermediate lemma as follows. ${\cal R}$ receives as input a state $\ket{\phi}$ and oracle access to ${\cal O}$. First, it runs $\adversary$ on input $\ket{\phi}$ while giving it oracle access to ${\cal O}$. $\adversary$ then outputs $k+1$ elements. ${\cal R}$ inputs each of the $k+1$ elements to the oracle ${\cal O}$. If the output of ${\cal O}$ on each of the $k+1$ distinct elements is 1 (i.e. the membership test passed) then output 1, else output 0. There are two cases to consider here: 
\begin{itemize}
    \item In the first case, $\ket{\phi}$ is $k$ copies of a random subset state and moreover, ${\cal O}$ tests membership in $S$. In this case, the probability that ${\cal R}$ outputs 1 is exactly $p$. 
    \item In the second case, $\ket{\phi}$ is a state of the form $\ket{\sigma_{{\bf x}}}$ and moreover, ${\cal O}$ tests membership in $\{x_1,\ldots,x_k\}$, where ${\bf x}=(x_1,\ldots,x_k)$. In this case, the probability that ${\cal R}$ outputs 1 is 0. This is because, $\{x_1,\ldots,x_k\}$ only has $k$ distinct elements. 
\end{itemize}
Thus, the success probability of ${\cal R}$ in violating the above intermediate lemma is $p$, which is non-negligible. This in turn is a contradiction.

\paragraph{Revocable Encryption.} Armed with the $k \rightarrow k+1$ unforgeability property, we will now tackle the first application of revocable encryption. Let us first formally state the syntax of such a scheme:
\begin{itemize}
\item $\KeyGen(1^\lambda)$: on input the security parameter $1^\lambda$, output a secret key $\sk$. 

\item $\Enc(\sk,m)$: on input the secret key $\sk$ and a message $m$, output a (pure) ciphertext state $\ket{\psi}$ and a (private) verification key $\vk$. 

\item $\Dec(\sk,\rho)$: on input the secret key $\sk$ and a quantum state $\rho$, output a message $m'$.

\item $\Revoke(\sk,\vk,\sigma)$: on input the secret key $\sk$, a verification key $\vk$ and a state $\sigma$, output $\top$ or $\bot$.
\end{itemize}
For multi-copy revocable security, we will consider the following experiment between a QPT adversary and a challenger.
\begin{enumerate}
    \item $\adversary$ submits two messages $m_0,m_1$ and a polynomial $k=k(\lambda)$ to the challenger.
    \item The challenger samples 
a key $\sk \leftarrow \KeyGen(1^{\secparam})$ and produces $\ket{\psi_b} \leftarrow \enc(\sk,m_b)$. Afterwards, the challenger sends the quantum state $\ket{\psi_b}^{\otimes k}$ to $\adversary$. 
    \item $\adversary$ returns a quantum state $\rho$. 
    \item The challenger performs the measurement $\left\{ \ketbra{\psi_b}{\psi_b}^{\otimes k},\ \mathbb{I} - \ketbra{\psi_b}{\psi_b}^{\otimes k} \right\}$ on the returned state $\rho$. If the measurement succeeds, the game continues; otherwise, the challenger aborts.
    \item The challenger sends the secret key $\sk$ to $\adversary$. 
    \item $\adversary$ outputs a bit $b'$. 
\end{enumerate}
We say that the revocable encryption scheme  has multi-copy revocable security, if no QPT adversary can guess $b'$ with probability that's more than negligibly better than $\frac{1}{2}.$

Now, we will describe how we get revocable encryption by making use of permutations. But before we explain our construction, we will begin by highlighting some challenges in achieving revocable encryption from subset states in the quantum random oracle model.\\

\noindent \textsc{Warm-up: Quantum Random Oracle Model.} We will design revocable encryption schemes in which part of the quantum ciphertext will consist of a random subset state. Since generating a random subset state involves first sampling a random exponential-sized set, such states cannot be efficiently generated. However, let us first see if random oracles can be exploited to generate random subset states more efficiently.

Let  $\mathcal{H}_{2n+m}= \{H:\{0,1\}^{2n+m}\rightarrow \{0,1\}^{m}\}$ be a hash function family, where $H$ is modeled as the random oracle. 
To encrypt a message $\mu \in \bit^m$, first sample a secret string $\mathrm{salt} \leftarrow \{0,1\}^{n+m}$ and coherently compute the following quantum state:
\[\ket{\psi}= 2^{-n/2}
\sum_{x \in \{0,1\}^n}\ket{x}\ket{H(\text{salt},x)}
\]
Now, we can measure the second register and obtain a random string $y$, in which case the state collapses to
\[
\ket{\psi_y} \propto
\sum_{x \in \{0,1\}^n \, : \, H(\text{salt}, x)=y}\ket{x}
\]
The ciphertext now consists of $\ket{\psi_y}$ together with $y \oplus \mu$. While this approach clearly solves the problem of efficiently generating the state, the construction has a major drawback: it only allows us to generate {\em one} copy of the state and, in particular, \emph{identical} copies of the subset state $\ket{\psi_y}$ are not efficiently preparable! Therefore, an entirely new approach is necessary in order to make multi-copy revocable encryption possible.\ \\

\noindent \textsc{Revocable Encryption Using Permutations.} In order to resolve the issue we encountered before, we instantiate our revocable encryption scheme using permutations instead. Concretely, we make use of quantum-secure pseudorandom permutations (QPRPs) which allow us to efficiently prepare many identical copies of subset states \emph{pseudorandomly}. Fortunately, by the security of the QPRP, the resulting states are effectively indisinguishable from random subset states. We then prove an oracular notion of revocable security: this means that, rather than revealing the QPRP key in the final part of the revocable security experiment, we instead allow the adversary to query an oracle for the permutation instead.\footnote{Note that the switch from a QPRP to a random permutation is not possible in the standard notion of revocable security in which the QPRP key is required to be revealed in the clear.}

Let $\Phi = \{\Phi_\lambda\}_{\lambda \in \N}$ be an ensemble of QPRPs $\Phi_\lambda = \{\varphi_\kappa : \bit^{n+m} \rightarrow \bit^{n+m}\}_{\kappa \in \algo K_\lambda}$, for some set $\algo K_\lambda$. Consider the scheme $\Sigma^\Phi=(\KeyGen,\Enc,\Dec,\Revoke)$ which consists of the following QPT algorithms:
 \begin{itemize}
    \item $\KeyGen(1^{\secparam})$: sample a uniformly random key $\kappa \in \algo K_\lambda$ and let $\sk = \kappa$.
    \item $\enc(\sk,\mu)$: on input the secret key $\sk=\kappa$ and message $\mu \in \bit^m$, sample  $y \sim \bit^m$ and prepare the subset state given by
    $$
\ket{S_y} = 
\frac{1}{\sqrt{2^n}}\sum_{x \in \bit^n} \ket{\varphi_\kappa(x || y)}. 
$$
    Output the ciphertext state $(\ket{S_y}^{\otimes k},y \oplus \mu)$ and (private) verification key $\vk=y$. 

    \item $\dec(\sk,\ct)$: on input the decryption key $\kappa$ and ciphertext state $(\ket{S_y},z) \leftarrow \ct$, coherently apply the in-place permutation $\phi^{-1}_{\kappa}$, measure to obtain a string $(x'||y')$, and output $y'\oplus z$. 
   \item $\Revoke(\sk,\vk,\rho)$: on input $\sk$, a state $\rho$ and verification key $\vk$, it parses $\kappa \leftarrow \sk$, $y \leftarrow \vk$ and applies the measurement $\left\{ \ketbra{S_y}{S_y},\ \mathbb{I} - \ketbra{S_y}{S_y} \right\}$ to $\rho$; it outputs $\top$ if it succeeds, and $\bot$ otherwise.
\end{itemize}


\noindent Observe that the state $\ket{S_y}$ can be efficiently generated using a unitary (that is part of the quantum ciphertext). Formally, we perform the following steps: 
\begin{itemize}
    \item Prepare a uniform superposition $\frac{1}{\sqrt{2^n}} \underset{x \in \{0,1\}^n}{\sum} \ket{x}$.
    \item Then, append $\ket{y}$ to the above state, where $y$ is sampled uniformly at random. 
    \item Finally, compute a unitary that applies the permutation $\varphi_\kappa$ in-place on each term in the state. 
\end{itemize}
\noindent This in turn means that many copies of $\ket{S_y}$ can be prepared efficiently. Moreover, using this unitary, the algorithm $\Revoke$ can also be efficiently implemented. Note that, by the security of the QPRP, the resulting state $\ket{S_y}$ is computationally indistinguishable from a variant of the state which is generated using perfectly random permutation.
\\
\\
\noindent \textsc{Proof Outline}:
Now, we will give a proof outline for proving multi-copy secure revocable encryption. 
First, we replace the QPRP by a perfectly random permutation. Second, we complete the proof via a reduction from the $k \rightarrow k+1$ unforgeability of random subset states. 
\par Let us assume for contradiction that there exists an adversary $\cal{A}$ such that the revocable encryption scheme presented above is not secure. The main idea behind converting $\cal{A}$ into an adversary $\cal{A'}$ against the $k \rightarrow k+1$ unforgeability of $\ket{S}$ is to construct hybrids that can be distinguished with non negligible probability, and then use the one-way to hiding lemma~\cite{cryptoeprint:2018/904} in order to construct an \emph{extractor} $\cal{E}$ which queries the membership oracle for $S$ on a $(k+1)^{st}$ distinct element \emph{after} $\cal{A}$ passes the revocation phase. Then, $\cal{A'}$ can extract $k+1$ elements given only $k$ copies of $\ket{S}$ as follows:
\begin{enumerate}
    \item $\cal{A}'$ runs $\cal{A}$ on input $\ket{S}^{\otimes k}$ and measures the registers returned by $\cal{A}$ that are supposed to contain $k$ copies of the subset state. With high probability (provided that the registers pass revocation), these measurements will result in precisely $k$ distinct elements in $S$.
    \item Then, $\mathcal{A'}$ can run the extractor $\cal{E}$ to extract the $(k+1)^{st}$ distinct element in $S$. 
\end{enumerate}
Note that extra care is needed to argue that the two steps above succeed simultaneously with sufficiently high probability whenever $\cal{A}$ is a successful adversary for the revocable security experiment. For simplicity, we omit these details in this overview and refer the reader to \Cref{sec:distinct-ext}.

Now, we will briefly describe the sequence of hybrids which helps us construct the extractor.\ \\

\noindent \textsc{Putting it All Together.}
 We are now ready to describe the sequence of hybrids we consider in order to prove the security of our construction in an oracular security model. The experiment is as follows:\\

\noindent \underline{$\revokeexperiment^{\adversary}(1^{\secparam},b)$:}
\begin{enumerate}
 \item $\adversary$ submits two $m$-bit messages $(\mu_0,\mu_1)$ and a polynomial $k=k(\secparam)$ to the challenger.
    \item The challenger samples $y \sim \bit^m$ and produces a quantum state
       $$
\ket{S_y} = 
\frac{1}{\sqrt{2^n}}\sum_{x \in \bit^n} \ket{\varphi(x || y)}. 
$$
The challenger then sends $\ket{S_y}^{\otimes k}$ and $y \oplus \mu_b$ to $\adversary$. 
    \item $\adversary$ prepares a bipartite state on registers $\mathsf{R}$ and $\mathsf{AUX}$, and sends $\mathsf{R}$ to the challenger and $\mathsf{AUX}$ to $\adversary$.
    \item The challenger performs the projective measurement $\left\{ \ketbra{S_y}{S_y}^{\otimes k},\ \mathbb{I} - \ketbra{S_y}{S_y}^{\otimes k} \right\}$ on $\mathsf{R}$. If the measurement succeeds, the challenger outputs $\bot$. Otherwise, the challenger continues.
    \item The challenger grants $\adversary$ quantum oracle access to $\varphi^{-1}$. 
    \item $\adversary$ outputs a bit $b'$. 
\end{enumerate}
We show that the following via hybrid experiments:
\begin{enumerate}
    \item  $\revokeexperiment^{\adversary}(1^{\secparam},b)$ is indistinguishable from the experiment $H_2$ which is the same as $\revokeexperiment^{\adversary}(1^{\secparam},b)$ except, \begin{itemize}
        \item The challenger samples a random subset $S \subseteq \bit^{n+m}$ of size $|S|=2^n$.
    \item The challenger samples a random $y \sim \bit^m$.
    \end{itemize} 
    after passing revocation, $\cal{A}$ receives access to the following oracle:
     $f: \bit^{n+m} \rightarrow \bit^{n+m}$ subject to the constraint that $f(s) = \ast||y$, for all $s \in S$. 
     \item $H_2$ is indistinguishable from the experiment $H_3$ which is the same as $H_2$ except,  after passing revocation, $\cal{A}$ receives access to the following oracle:
     $f: \bit^{n+m} \rightarrow \bit^{n+m}$ subject to the constraint that $f(s) = \ast||u$, for all $s \in S$. Here, $u$ is a randomly sampled string. 
\end{enumerate}

With some simple hybrid arguments, we show that if $\revokeexperiment^{\adversary}(1^{\secparam},b)$ is not secure, then $\cal{A}$ can distinguish $H_2$ and $H_3$. However, since $H_2$ and $H_3$ are the same except the the function $f$ differs \emph{only} on inputs in the subset $S$, we can now use one-way to hiding to come up with an extractor $\cal{E}$ extracts an element in the subset $S$ using the  queries $\cal{A}$ makes to  $f$ \emph{after revocation succeeds}. 
To complete the proof, we construct an adversary that breaks the $k \rightarrow  k+1$ unforgeability of subset states, in violation of our main lemma.
The high level idea is that the $k$ copies of the subset state which are meant to be returned easily allow the adversary to obtain $k$ distinct subset elements; whereas the additional $(k+1)$-st element can be extracted from the adversary's side information via the aforementioned extractor.\footnote{Note that there are a few subtleties that we gloss over in this overview, such as the fact that we need to condition on the event that revocation succeeds. We refer the reader to \Cref{sec:revenc} for the full proof.} 

\paragraph{Revocable Programs.}
In addition to revocable encryption, we also achieve multi-copy secure revocable programs in the \emph{classical} oracle model from the $k\rightarrow k+1$ unforgeability of subset states. 
The proof of security in this setting is similar, and the main idea again is to construct an extractor extracting $k+1$ elements given only $k$ copies of $S$, though extra care is required when defining the hybrids in order to prove security. For the formal proof, we refer the reader to  \Cref{sec:revocavle programs-def}.

\paragraph{Revocable Point Functions.}
As a final application of $k \rightarrow k+1$ unforgeability, we give a revocable multi-copy secure scheme for a particular program; namely, multi-bit point functions $P_{y,m}$ of the form
$$ P_{y,m} (x) = \begin{cases} m   & \text{if } x = y\,,\\
    0^\lambda &\text{if } x \neq y \,,   \end{cases}  $$ 
    in the quantum random oracle model.
The main idea is to show that revocable encryption schemes can be generically converted to schemes which satisfy a property called \emph{wrong-key detection} (WKD); broadly speaking, this is a mechanism that allows a user with access to ciphertext to check whether a particular of choice decryption key is incorrect, in which case decryption algorithm simply outputs $\bot$.
Inspired by~\cite{Coladangelo2024quantumcopy}, we give a transformation that converts any revocable encryption scheme to one which satisfies the WKD property. Intuitively, this achieves our goal: to evaluate $P_{y,m}$ at point $x$, attempt to decrypt using $x$; if decryption succeeds output
the decrypted message, if decryption fails, output $0^{\lambda}$ (see \Cref{sec:pointfunction}).

\section{Preliminaries} 
\noindent Let $\secparam \in \N$ denote the security parameter throughout this work. We assume that the reader is familiar with the fundamental cryptographic concepts. 

For $N\in \N$, we use $[N] = \{1,2,\dots,N\}$ to denote the set of integers up to $N$. The symmetric group on $[N]$ is denoted by $S_N$. 
In slight abuse of notation, we oftentimes identify elements $x \in [N]$ with bit strings $x \in \bit^n$ via their binary representation whenever $N=2^n$ and $n \in \N$. Similarly, we identify permutations $\pi \in S_N$ with permutations $\pi: \bit^{n} \rightarrow 
\bit^n$ over bit strings of length $n$. For a bit string $x \in \bit^n$, we frequently use the notation $(x||*)$, where $*$ serves as a placeholder to denote the set $\{(x||y) \, : \, y \in \bit^m \}$, where $m \in \N$ is another integer which is typically clear in context.

We write $\negl(\cdot)$ to denote any \emph{negligible} function, which is a function $f$ such that, for every constant $c \in \mathbb{N}$, there exists an integer $N$ such that for all $n > N$, $f(n) < n^{-c}$.

\paragraph{Quantum Computing} For a comprehensive background on quantum computation, we refer to \cite{NielsenChuang11}. We denote a finite-dimensional complex Hilbert space by $\mathcal{H}$, and we use subscripts to distinguish between different systems (or registers). For example, we let $\mathcal{H}_{A}$ be the Hilbert space corresponding to a system $A$. 
The tensor product of two Hilbert spaces $\algo H_A$ and $\algo H_B$ is another Hilbert space denoted by $\algo H_{AB} = \algo H_A \otimes \algo H_B$.
The Euclidean norm of a vector $\ket{\psi} \in \algo H$ over the finite-dimensional complex Hilbert space $\mathcal{H}$ is denoted as $\| \psi \| = \sqrt{\braket{\psi|\psi}}$. 
Let $\algo L(\algo H)$
denote the set of linear operators over $\algo H$. A quantum system over the $2$-dimensional Hilbert space $\mathcal{H} = \mathbb{C}^2$ is called a \emph{qubit}. For $n \in \mathbb{N}$, we refer to quantum registers over the Hilbert space $\mathcal{H} = \big(\mathbb{C}^2\big)^{\otimes n}$ as $n$-qubit states. We use the word \emph{quantum state} to refer to both pure states (unit vectors $\ket{\psi} \in \mathcal{H}$) and density matrices $\rho \in \mathcal{D}(\mathcal{H)}$, where we use the notation $\mathcal{D}(\mathcal{H)}$ to refer to the space of positive semidefinite matrices of unit trace acting on $\algo H$. 
The \emph{trace distance} of two density matrices $\rho,\sigma \in \mathcal{D}(\mathcal{H)}$ is given by
$$
\TD(\rho,\sigma) = \frac{1}{2} \Tr\left[ \sqrt{ (\rho - \sigma)^\dag (\rho - \sigma)}\right].
$$
A quantum channel $\Phi: \algo L(\algo H_A) \rightarrow \algo L(\algo H_B)$ is a linear map between linear operators over the Hilbert spaces $\algo H_A$ and $\algo H_B$. Oftentimes, we use the compact notation $\Phi_{A \rightarrow B}$ to denote a quantum channel between $\algo L(\algo H_A)$ and $\algo L(\algo H_B)$. We say that a channel $\Phi$ is \emph{completely positive} if, for a reference system $R$ of arbitrary size, the induced map $I_R \otimes \Phi$ is positive, and we call it \emph{trace-preserving} if $\Tr[\Phi(X)] = \Tr[X]$, for all $X \in \algo L(\algo H)$. A quantum channel that is both completely positive and trace-preserving is called a quantum $\mathsf{CPTP}$ channel.
A \emph{unitary} $U: L(\mathcal{H}_A) \rightarrow L(\mathcal{H}_A)$ is a special case of a quantum channel that satisfies $U^\dagger U = U U^\dagger = I_A$. An isometry is a linear map $V: L(\mathcal{H}_A) \rightarrow L(\mathcal{H}_B)$ with $\dim(\mathcal{H}_B) \geq \dim(\mathcal{H}_A)$ and $V^\dag V = I_A$.
A \emph{projector} ${\Pi}$ is a Hermitian operator such that ${\Pi}^2 = {\Pi}$, and a \emph{projective measurement} is a collection of projectors $\{{\Pi}_i\}_i$ such that $\sum_i {\Pi}_i = I$.
A positive-operator valued measure ($\mathsf{POVM}$) is a set of Hermitian positive semidefinite operators $\{M_i\}$ acting on a Hilbert space $\mathcal{H}$ such that $\sum_{i} M_i = I$. 


\paragraph{Quantum algorithms.}
A polynomial-time \emph{uniform} quantum algorithm (or $\mathsf{QPT}$ algorithm) is a polynomial-time family of quantum circuits given by $\mathcal{C} = \{C_\lambda\}_{\lambda \in \N}$, where each circuit $C \in \algo C$ is described by a sequence of unitary gates and measurements; moreover, for each $\lambda \in \N$, there exists a deterministic polynomial-time Turing machine that, on input $1^\lambda$, outputs a circuit description of $C_\lambda$. Similarly, we also define (classical) probabilistic polynomial-time $(\mathsf{PPT})$ algorithms. A quantum algorithm may, in general, receive (mixed) quantum states as inputs and produce (mixed) quantum states as outputs. We frequently restrict $\mathsf{QPT}$ algorithms implicitly; for example, if we write $\Pr[\mathcal{A}(1^{\lambda}) = 1]$ for a $\mathsf{QPT}$ algorithm $\mathcal{A}$, it is implicit that $\mathcal{A}$ is a $\mathsf{QPT}$ algorithm that outputs a single classical bit.
We say that a quantum algorithm $\mathcal{A}$ has oracle access to a classical function $f: \{0,1 \}^{n} \rightarrow \{0,1 \}^m$, denoted by $\mathcal{A}^f$, if $\mathcal{A}$ is allowed to use a unitary gate $\oracle_f$ at unit cost in time. The unitary $\oracle_f$ acts as follows on the computational basis states of a Hilbert space $\mathcal{H}_X \otimes \mathcal{H}_Y$ of $n+m$ qubits:
$$
\oracle_f: \quad
\ket{x}_X \otimes \ket{y}_Y \longrightarrow \ket{x}_X \otimes \ket{y \oplus f(x)}_Y,
$$
where the operation $\oplus$ denotes bit-wise addition modulo $2$. Oracles with quantum query-access have been studied extensively, for example in the context of quantum complexity theory~\cite{Bennett_1997}, as well as in cryptography~\cite{10.1007/978-3-642-25385-0_3,cryptoeprint:2018/904,cryptography4010010}.

\paragraph{One-Way-to-Hiding Lemma.}

We use the following lemma which combines~\cite[Theorem 3]{cryptoeprint:2018/904} and ~\cite[Lemma 8]{cryptoeprint:2018/904}, where the latter is rooted in Vazirani's \emph{Swapping Lemma}~\cite{Vaz98}. 

\begin{lemma}[One-Way-to-Hiding Lemma,~\cite{cryptoeprint:2018/904}]\label{lem:O2H}
Let $\algo X,\algo Y$ be arbitrary sets and let $\algo S \subseteq \algo X$ be a (possibly random) subset. Let $G,H: \algo X \rightarrow \algo Y$ be arbitrary (possibly random) functions such that $H(x)=G(x)$, for all $x \notin \algo S$. Let $z$ be a classical bit string or a (possibly mixed) quantum state (Note that $G,H,S,z$ may have arbitrary joint distribution). Let $\algo A$ be an
oracle-aided quantum algorithm that makes at most $q$ quantum queries. Let $\algo B$ be an algorithm that on input $z$ chooses a random query index $i \leftarrow [q]$, runs $\algo A^H(z)$, measures $\algo A$'s $i$-th query and outputs the measurement
outcome. Then, we have    
$$
\left|\Pr[\algo A^G(z)=1] - \Pr[\algo A^H(z)=1] \right| \leq 2 q \sqrt{\Pr[\algo B^H(z) \in \algo S]}.
$$
Moreover, for any fixed choice of $G,H,S$ and $z$ (when $z$ is a classical string or a pure state), we get
$$
\big\| \ket{\psi^H_q} - \ket{\psi^G_q} \big\| \leq 2q\sqrt{\frac{1}{q}\sum_{i = 0}^{q-1} \big\|\Pi_{\algo S}\ket{\psi^H_i}\big\|^2},
$$
where $\ket{\psi^H_i}$ denotes the intermediate state of $\algo A$ just before the $(i+1)$-st query, where the initial state at $i=0$ corresponds to $z$, and $\Pi_{\algo S}$ is a projector onto $\algo S$.
\end{lemma}

\paragraph{Pseudorandom Permutations.}

A quantum-secure pseudorandom permutation is a a bijective function family which can be constructed from quantum-secure one-way functions~\cite{zhandry2016notequantumsecureprps}.
\begin{definition}[QPRP]\label{def:qprp}
Let $\lambda \in \N$ denote the security parameter. Let $P: \bit^\lambda \times \bit^n \rightarrow \bit^n$ be a function, where $n(\lambda)=\poly(\lambda)$ is an integer, such that each function $P_k(x)=P(k,x)$ in the corresponding family $\{P_k\}_{k\in\bit^\lambda}$ is bijective. We say $P$ is a (strong) \textit{quantum-secure pseudorandom permutation} (or QPRP) if, for every QPT $\algo A$ with access to both the function and its inverse, it holds that
\begin{equation*}\label{eq:PRP}
\left|\Pr_{k \sim \bit^\lambda} \left[\algo A^{P_k,P_k^{-1}}(1^\lambda) = 1\right]
- \Pr_{\varphi\sim \mathcal{P}_n} \left[\algo A^{\varphi,\varphi^{-1}}(1^\lambda) = 1\right]\right| \leq \negl(\lambda)\,,
\end{equation*}
where $\mathcal P_n$ denotes the set of permutations over $n$-bit strings.
\end{definition}

We use the following result due to Zhandry which applies to random permutations.

\begin{theorem}[\cite{zhandry2013notequantumcollisionset}, Theorem 3.1] \label{thm:Zhandry-result} Any $q$-query quantum algorithm can distinguish random functions from random permutations over $\bit^n$ with advantage at most $O(q^3/2^n)$.
\end{theorem}

\paragraph{Subset States.} We consider the following notations. 
\begin{itemize}

\item We denote the set of distinct $k$-tuples over a set $S$ by $\mathrm{dist}(S,k)$.
    \item Suppose $S$ is a set. We denote $\ket{S} = \frac{1}{\sqrt{|S|}} \sum_{x \in S} \ket{x}$. 
    \item Suppose $X=\{x_1,\ldots,x_t\} \subseteq \{0,1\}^{n}$. We denote $\ket{\perm_{X}} = \frac{1}{t!} \sum_{\sigma \in S_t}  \ket{x_{\sigma(1)},\ldots,x_{\sigma(t)}}$, where $S_t$ denotes the symmetric group on $[t]$.  
    
\end{itemize}

We use the following lemma which follows from Propositions 3.3 and 3.4 in~\cite{JMW23}.

\begin{lemma}[\cite{JMW23}]
\label{lem:jmw23}
Let $n,k \in \N$. Let $T \subseteq \bit^n$ be a subset of size $|T|=t$. Then, it holds that
$${\sf TD}\left( \underset{\substack{S \subseteq T\\
|S|=s}}{\E} \left[ \ketbra{S}{S}^{\otimes k} \right],\ \underset{\substack{X \subseteq T\\
|X|=k}}{\E} \big[ \ketbra{\perm_{X}}{\perm_X} \big]  \right) \, \leq \, O\left( \frac{k}{\sqrt{s}}  + \frac{s k}{t} \right).$$    
\end{lemma}

\section{$k \mapsto k+1$ Unforgeability of Subset States} 

We now prove the following theorem. Roughly speaking, our theorem says that any quantum algorithm which receives $k$ copies of a random subset state $\ket{S}$ (and a membership oracle for $S$) cannot find $k+1$ distinct elements in $S$ with high probability unless it makes a large number of queries.

\begin{theorem}[$k \mapsto k+1$ Unforgeability of Subset States]\label{thm:rss:pd}
Let $n \in \N$ and $k \in \N$. Then, for any $q$-query quantum oracle algorithm $\mathcal{A}$, and any $1 \leq k < s < t \leq 2^n$, it holds that
\begin{align*}
&\Pr_{\substack{S \subseteq \bit^n\\
|S|=s}}\Big[
(x_1,\dots,x_{k+1}) \in \mathrm{dist}(S,k+1) \, : \, (x_1,\dots,x_{k+1})\leftarrow \mathcal{A}^{\oracle_S}(\ket{S}^{\otimes k}) 
\Big]\\
&\quad \leq O\left(q \cdot \sqrt{\frac{t-s}{2^n}}+q \cdot \sqrt{\frac{t-k}{2^n}} + \frac{k}{\sqrt{s}}  + \frac{s k}{t}\right)+\negl(n).
\end{align*}
In particular, we can let $k=\poly(n)$, $q=\poly(n)$ and $s(n)= n^{\omega(1)}$ be superpolynomial. Then, for any $t(n)= n^{\omega(1)}$ with $s(n)/t(n) = 1/n^{\omega(1)}$ and $t(n)/2^n = 1/n^{\omega(1)}$, the probability is at most $\negl(n)$.

\end{theorem}
\begin{proof}
Using Lemma \ref{lem:rss:pdf:helpful}, we can prove Theorem \ref{thm:rss:pd} as follows: Let's assume for contradiction that there exists a $q$ query quantum oracle algorithm $\cal{A}$, and a $1\leq k<s<t\leq 2^n$, such that 

 \begin{align*}
&\Pr_{\substack{S \subseteq \bit^n\\
|S|=s}}\Big[
(x_1,\dots,x_{k+1}) \in \mathrm{dist}(S,k+1) \, : \, (x_1,\dots,x_{k+1})\leftarrow \mathcal{A}^{\oracle_S}(\ket{S}^{\otimes k}) 
\Big]\\
&\quad = O\left(q \cdot \sqrt{\frac{t-s}{2^n}}+q \cdot \sqrt{\frac{t-k}{2^n}}+ \frac{k}{\sqrt{s}}  + \frac{s k}{t}\right)+\delta(n).
\end{align*}
where $\delta(.)$ is some non negligible function. 
Then, from Lemma \ref{lem:rss:pdf:helpful}, this implies that,
\begin{align*}
&\Pr_{\substack{
X \subseteq \{0,1\}^n, |X|=k}}\Big[
(x_1,\dots,x_{k+1}) \in \mathrm{dist}(X,k+1) \, : \, (x_1,\dots,x_{k+1})\leftarrow \mathcal{A}^{\oracle_X}(\ket{\perm_{X}}) 
\Big]\\&\geq \delta(n)
\end{align*}
However, the probability that $\cal{A}$ can succeed in this experiment is 0. This is because, $X$ only contains $k$ elements, and therefore, $\cal{A}$ can never produce $k+1$ distinct elements from $X$. This must imply that $\delta(n)$ is negligible. 

\end{proof}

\paragraph{Technical Lemma.}

We need to show the following lemma which made use of in \Cref{thm:rss:pd}.

\begin{lemma} 
\label{lem:rss:pdf:helpful}
Let $n,k,t \in \N$ be integers such that $1 \leq k < s < t \leq 2^n$. Then, for any $q$-query quantum oracle algorithm $\mathcal{D}$ which outputs a single bit, it holds that
\begin{align*}
\Bigg| &\Pr_{\substack{
S \subseteq \bit^n\\
|S|=s
}}\left[\mathcal{D}^{\oracle_S}\left( \ket{S}^{\otimes k}\right)=1 \right] - \Pr_{\substack{X \subseteq \bit^n\\|X|=k}}\left[\mathcal{D}^{\oracle_X}\left( \ket{\perm_{X}}\right)=1 \right] \Bigg|\\
&\leq O\left(q \cdot \sqrt{\frac{t-s}{2^n}} +q \cdot \sqrt{\frac{t-k}{2^n}} +\frac{k}{\sqrt{s}}  + \frac{s k}{t} \right).
\end{align*}
\end{lemma} 

\begin{proof} Consider the following hybrid distributions. \\

\noindent $\hybrid_1$: Output $\mathcal{D}^{\oracle_S}( \ket{S}^{\otimes k})$, where $S \subseteq \bit^n$ is a random subset of size
$|S|=s$.\\

\noindent $\hybrid_2$: Output $\mathcal{D}^{\oracle_S}( \ket{S}^{\otimes k})$, where the subset $S$ is sampled as follows: first, sample a random subset $T \subseteq \bit^n$ of size $|T|=t$, and then let $S \subseteq T$ be a random subset of size
$|S|=s$.\\

\par Let $p(\hybrid_i)$ be the probability that $\hybrid_i$ outputs 1, for some $i$. We now show the following. 
\begin{claim}
$p(\hybrid_2) = p(\hybrid_1)$. 
\end{claim}
\begin{proof}
The distribution of sampling $S \subseteq \bit^n$ of size
$|S|=s$ is identical to the distribution of first sampling a superset $T \subseteq \bit^n$ of size $|T|=t$, and letting $S \subseteq T$ be a random subset of size
$|S|=s$.
\end{proof}

\noindent $\hybrid_3$: Output $\mathcal{D}^{\oracle_T}( \ket{S}^{\otimes k})$, where the subset $S$ is sampled as follows: first, sample a random subset $T \subseteq \bit^n$ of size $|T|=t$, and then let $S \subseteq T$ be a random subset of size
$|S|=s$.

\begin{claim}
$$
\left| p(\hybrid_3) - p(\hybrid_2) \right| \leq O\left(q \cdot \sqrt{\frac{t-s}{2^n}}\right). 
$$
\end{claim}
\begin{proof}
We can model the quantum oracle algorithm $\mathcal{D}^{\oracle_S}$ on input $\ket{S}^{\otimes t}$ as a sequence of oracle queries and unitary computations followed by a measurement. Thus, the final output state just before the measurement can be written as
$$
\ket{\Psi_q^S} = U_q \oracle_S U_{q-1} \dots U_1 \oracle_S U_0\ket{\psi_0}\ket{S}^{\otimes k} \,, 
$$
where $U_0,U_1,\dots,U_q$ are unitaries (possibly acting on additional workspace registers, which we omit above), and where $\ket{\psi_0}$ is some fixed initial state which is independent of $S$. 

In the next step of the proof, we will use the ``subset flooding'' technique to drown $S$ in a random superset.
Let $T \subseteq \bit^n$ be a random superset of $S$ of size $t >s$. We now consider the state 
$$
\ket{\Psi_q^T} = U_q \oracle_T U_{q-1} \dots U_1 \oracle_T U_0\ket{\psi_0}\ket{S}^{\otimes k}.
$$
We now claim that the states $\ket{\Psi_q^S}$ and $\ket{\Psi_q^T}$ are sufficiently close.
From the definition of $\oracle_T$ and $\oracle_S$, we have that $\oracle_T(x)\neq \oracle_S(x)$ iff $x \in T\backslash S \subset \{0,1\}^{n}$. By the O2H Lemma (\Cref{lem:O2H}), 
\begin{align*}
\E_{\substack{T \subseteq \bit^n, |T|=t\\
S \subseteq T, |S|=s}}\left\| \ket{\Psi_q^S} - \ket{\Psi_q^T} \right\| 
&\leq  2q\E_{\substack{T \subseteq \bit^n, |T|=t\\
S \subseteq T, |S|=s}}\sqrt{ \frac{1}{q} \sum_{i=0}^{q-1}\big\|\Pi_{T\backslash S} \ket{\Psi_{i}^S}\big\|^2}\\
&\leq 2q\sqrt{
\frac{1}{q}\sum_{i=0}^{q-1} \,\E_{\substack{T \subseteq \bit^n, |T|=t\\
S \subseteq T, |S|=s}} \big\|\Pi_{T\backslash S} \ket{\Psi_{i}^S}\big\|^2} & \text{(Jensen's inequality)}\\
&= O\left(q \cdot \sqrt{\frac{t-s}{2^n}}\right).
\end{align*}
Therefore, the probability (over the choice of $S$ and $T$) that $\mathcal{D}^{\oracle_S}(\ket{S}^{\otimes k})$ succeeds is 
 at most the probability that $\mathcal{D}^{\oracle_T}(\ket{S}^{\otimes k})$ succeeds---up to an additive loss of
$O(q \cdot \sqrt{\frac{t-s}{2^n}})$. 

\end{proof}

\noindent $\hybrid_4$: Output $\mathcal{D}^{\oracle_T}(\ket{\perm_{X} })$, where the subset $X$ is sampled as follows: first, sample a random subset $T \subseteq \bit^n$ of size $|T|=t$, and then let $X \subseteq T$ be a random subset of size
$|X|=k$.

\begin{claim}
$$
\left| p(\hybrid_4) - p(\hybrid_3) \right| \leq O\left( \frac{k}{\sqrt{s}}  + \frac{s k}{t} \right). 
$$
\end{claim}
\begin{proof}
Here, we make use of~\Cref{lem:jmw23} which says that, for any superset $T \subseteq \bit^n$ of size $|T|=t$,
$${\sf TD}\left( \underset{\substack{S \subseteq T\\
|S|=s}}{\E} \left[ \ketbra{S}{S}^{\otimes k} \right],\ \underset{\substack{X \subseteq T\\
|X|=k}}{\E} \big[ \ketbra{\perm_{X}}{\perm_X} \big]  \right) \, \leq \, O\left( \frac{k}{\sqrt{s}}  + \frac{s k}{t} \right).$$ 

\end{proof}

\noindent $\hybrid_5$: Output $\mathcal{D}^{\oracle_X}(\ket{\perm_{X} })$, where the subset $X$ is sampled as follows: first, sample a random subset $T \subseteq \bit^n$ of size $|T|=t$, and then let $X \subseteq T$ be a random subset of size
$|X|=k$. 

\begin{claim}
$$
\left| p(\hybrid_5) - p(\hybrid_4) \right| \leq O\left(q \cdot \sqrt{\frac{t-k}{2^n}}\right). 
$$
\end{claim}

\begin{proof}
Suppose that $X \subseteq T$ is a random subset of size $|X|=k$, and let $\ket{\perm_{X}} = \frac{1}{k!} \sum_{\sigma \in S_k}  \ket{x_{\sigma(1)},\ldots,x_{\sigma(k)}}$, where $S_k$ denotes the symmetric group on $[k]$. Consider the state
$$
\ket{\Phi_q^T} = U_q \oracle_T U_{q-1} \dots U_1 \oracle_T U_0\ket{\psi_0}\ket{\perm_{X}}.
$$
prepared by $\mathcal{D}^{\oracle_T}(\ket{\perm_X})$ just before the measurement. Similarly, we let
$$
\ket{\Phi_q^X} = U_q \oracle_X U_{q-1} \dots U_1 \oracle_X U_0\ket{\psi_0}\ket{\perm_{X}}.
$$
be the state prepared by $\mathcal{A}^{\oracle_X}(\ket{\perm_X})$. Using \Cref{lem:O2H} as before, we get that
\begin{align*}
\E_{\substack{T \subseteq \bit^n, |T|=t\\
X \subseteq T, |X|=k}}\left\| \ket{\Phi_q^T} - \ket{\Phi_q^X} \right\|
&\leq  2q\E_{\substack{T \subseteq \bit^n, |T|=t\\
X \subseteq T, |X|=k}}\sqrt{
\frac{1}{q}\sum_{i=0}^{q-1}\big\|\Pi_{T\backslash X} \ket{\Phi_{i}^X}\big\|^2}\\
&\leq 2q\sqrt{
\frac{1}{q} \sum_{i=0}^{q-1} \,\E_{\substack{T \subseteq \bit^n, |T|=t\\
X \subseteq T, |X|=k}} \big\|\Pi_{T\backslash X} \ket{\Phi_{i}^X}\big\|^2} & \text{(Jensen's inequality)}\\
&\leq O\left(q \cdot \sqrt{\frac{t-k}{2^n}} \right).
\end{align*}
Therefore, the probability (over the choice of $X$ and $T$) that $\mathcal{D}^{\oracle_T}(\ket{\perm_{X}})$ succeeds is 
 at most the probability that $\mathcal{D}^{\oracle_X}(\ket{\perm_{X}})$ succeeds (up to an additive error of
$\frac{q}{\sqrt{t-k}}$).  
\end{proof}
Therefore, by applying the triangle inequality, we get that
$$
\left| p(\hybrid_1) - p(\hybrid_5) \right| \leq O\left(q \cdot \sqrt{\frac{t-s}{2^n}} +q \cdot \sqrt{\frac{t-k}{2^n}} +\frac{k}{\sqrt{s}}  + \frac{s k}{t} \right).
$$
This proves the claim.
\end{proof}

\section{Multi-Copy Revocable Encryption: Definition}
\label{sec:revenc}

In this section we formally define and construct multi-copy secure revocable encryption schemes. These are regular encryption scheme but where the ciphertexts are associated with quantum states. Moreover, the security property guarantees that any adversary that successfully returns $k$ valid copies of a quantum ciphertext (which it received from a trusted party), where $k$ is an arbitrary polynomial, necessarily loses the ability to decrypt the ciphertexts in the future---even if the secret key is revealed.

Our definition of revocable encryption is as follows:

\begin{definition}[Revocable Encryption] Let $\lambda \in \N$ denote the security parameter. A revocable encryption scheme $\Sigma=(\KeyGen,\Enc,\Dec,\Revoke)$ with plaintext space $\algo M$ consists of the following QPT algorithms:
\begin{itemize}
\item $\KeyGen(1^\lambda)$: on input the security parameter $1^\lambda$, output a secret key $\sk$. 

\item $\Enc(\sk,m)$: on input the secret key $\sk$ and a message $m \in \algo M$, output a (pure) ciphertext state $\ket{\psi}$ and a (private) verification key $\vk$. 

\item $\Dec(\sk,\rho)$: on input the secret key $\sk$ and a quantum state $\rho$, output a message $m'$.

\item $\Revoke(\sk,\vk,\sigma)$: on input the secret key $\sk$, a verification key $\vk$ and a state $\sigma$, output $\top$ or $\bot$.
\end{itemize}
In addition, we require that $\Sigma$ satisfies the following two properties:
\begin{description}
    \item \textbf{Correctness of decryption:} for all plaintexts $m \in \algo M$, it holds that
$$\prob\left[m \leftarrow \dec(\sk,\ket{\psi}) \ :\ \substack{\sk \leftarrow \KeyGen(1^{\secparam})\\ \ \\ (\ket{\psi},\vk) \leftarrow \enc(\sk,m)} \right] \geq 1 - \negl(\secparam).$$

\item \textbf{Correctness of revocation:} for all plaintexts $m \in \algo M$, it holds that
$$\prob\left[\top \leftarrow \Revoke(\sk,\vk,\ket{\psi}) \ :\ \substack{\sk \leftarrow \KeyGen(1^{\secparam})\\ \ \\ (\ket{\psi},\vk) \leftarrow \enc(\sk,m)} \right] \geq 1 - \negl(\secparam).$$
\end{description}
\end{definition}

\noindent There are two properties we require the above scheme to satisfy. \\

\noindent Firstly, we require the above scheme to be correct. That is, we require the following to hold for all $m \in \{0,1\}^{\ell}$,
$$\prob\left[m \leftarrow \dec(\sk,\ket{\psi}) \ :\ \substack{\sk \leftarrow \KeyGen(1^{\secparam})\\ \ \\ \ket{\psi} \leftarrow \enc(\sk,m)} \right] \geq 1 - \eps(\secparam),$$
for some negligible function $\eps(\cdot)$. \\

\paragraph{Multi-Copy Revocable Security.}

We use the following notion of security.

\begin{definition}[Multi-Copy Revocable Security] Let $\lambda \in \N$ denote the security parameter and let $\Sigma=(\KeyGen,\Enc,\Dec,\Revoke)$ be a revocable encryption scheme with plaintext space $\algo M$. Consider the following experiment between a QPT adversary $\algo A$ and a challenger.\\
\ \\
\noindent \underline{$\revokeexperiment_{\lambda,\Sigma,\adversary}(b)$:}
\begin{enumerate}
    \item $\adversary$ submits two messages $m_0,m_1 \in \algo M$ and a polynomial $k=k(\lambda)$ to the challenger.
    \item The challenger samples 
a key $\sk \leftarrow \KeyGen(1^{\secparam})$ and produces $\ket{\psi_b} \leftarrow \enc(\sk,m_b)$. Afterwards, the challenger sends the quantum state $\ket{\psi_b}^{\otimes k}$ to $\adversary$. 
    \item $\adversary$ returns a quantum state $\rho$. 
    \item The challenger performs the measurement $\left\{ \ketbra{\psi_b}{\psi_b}^{\otimes k},\ \mathbb{I} - \ketbra{\psi_b}{\psi_b}^{\otimes k} \right\}$ on the returned state $\rho$. If the measurement succeeds, the game continues; otherwise, the challenger aborts.
    \item The challenger sends the secret key $\sk$ to $\adversary$. 
    \item $\adversary$ outputs a bit $b'$. 
\end{enumerate}
We say that the revocable encryption scheme $\Sigma=(\KeyGen,\Enc,\Dec,\Revoke)$ has multi-copy revocable security, if the following holds for all $\mu \in \{0,1,\bot\}$:
$$\left| \prob\big[ \mu \leftarrow \revokeexperiment_{\lambda,\Sigma,\adversary}(0) \big] - \prob\big[ \mu \leftarrow \revokeexperiment_{\lambda,\Sigma,\adversary}(1) \right] \big| \leq \negl(\secparam) \, ,
$$
\begin{remark}
    In this work, we will work with a weaker variant of this security game, one where the post revocation adversary is given access to an oracle that depends on the key $\mathsf{sk}$ instead. The reason for this will become clear from context.
\end{remark}

\begin{remark}[Search variant]\label{search-security}
Occasionally, we also consider the search variant of multi-copy revocable encryption. Here, the experiment is similar, except that in Step $1$, the adversary only submits $k$ to the challenger. Then, in Step $2$, the challenger chooses a message $m$ uniformly at random from the plaintext space, encrypts it $k$ times and sends all of the copies to the adversary. Finally, the adversary is said to win the game if it guesses $m$ correctly.
\end{remark}
\end{definition}

\section{Construction of Multi-Copy Secure Revocable Encryption}
\label{sec:revocable-encryption-constr}

In this section,
we instantiate our revocable encryption scheme using quantum-secure pseudorandom permutations (QPRPs), and we prove security in an oracular model: this means that, rather than revealing the QPRP key in the final part of the revocable security experiment, we instead allow the adversary to query an oracle for a the permutation instead.

\begin{construction}\label{const:permutation-scheme} Let $\lambda \in \N$ be the security parameter. Let $n,m \in \N$ be polymomial in $\lambda$. Let $\Phi = \{\Phi_\lambda\}_{\lambda \in \N}$ be an ensemble of permutations $\Phi_\lambda = \{\varphi_\kappa : \bit^{n+m} \rightarrow \bit^{n+m}\}_{\kappa \in \algo K_\lambda}$, for some set $\algo K_\lambda$. Consider the scheme $\Sigma^\Phi=(\KeyGen,\Enc,\Dec,\Revoke)$ which consists of the following QPT algorithms:
 \begin{itemize}
    \item $\KeyGen(1^{\secparam})$: sample a uniformly random key $\kappa \in \algo K_\lambda$ and let $\sk = \kappa$.
    \item $\enc(\sk,\mu)$: on input the secret key $\sk=\kappa$ and message $\mu \in \bit^m$, sample  $y \sim \bit^m$ and prepare the subset state given by
    $$
\ket{S_y} = 
\frac{1}{\sqrt{2^n}}\sum_{x \in \bit^n} \ket{\varphi_\kappa(x || y)}. 
$$
    Output the ciphertext state $(\ket{S_y}^{\otimes k},y \oplus \mu)$ and (private) verification key $\vk=y$. 
    
    \item $\dec(\sk,\ct)$: on input the decryption key $\kappa$ and ciphertext state $(\ket{S_y},z) \leftarrow \ct$, do the following: 
    \begin{itemize}
        \item Coherently compute $\varphi_\kappa^{-1}$ on $\ket{S_y}$ and store the answer in a separate output register. 
        \item Measure the output register to get $x'||y' \in \bit^{n+m}$. 
        \item Output $y' \oplus z$. 
    \end{itemize}
\item $\Revoke(\sk,\vk,\rho)$: on input $\sk$, a state $\rho$ and verification key $\vk$, it parses $\kappa \leftarrow \sk$, $y \leftarrow \vk$ and applies the measurement $\left\{ \ketbra{S_y}{S_y},\ \mathbb{I} - \ketbra{S_y}{S_y} \right\}$ to $\rho$; it outputs $\top$ if it succeeds, and $\bot$ otherwise.
\end{itemize}   
\end{construction}

\paragraph{Proof of Multi-Copy Revocable Security.}

\begin{theorem}\label{thm:revoc-encryption}

\Cref{const:permutation-scheme}, when instantiated with a QPRP $\Phi = \{\Phi_\lambda\}_{\lambda \in \N}$, satisfies (an oracular notion of) multi-copy revocable security.
\end{theorem}

\begin{proof}
Because we are working in the oracular model of revocable security, we can invoke QPRP security and assume that $\Sigma$ in \Cref{const:permutation-scheme} is instantiated with a perfectly random permutation $\varphi$ rather than a QPRP permutation $\varphi_\kappa$.

Suppose that our construction does not achieve multi-copy revocable security. Then, there exists an adversary $\adversary$ such that
$$\left| \prob\left[ \mu \leftarrow \revokeexperiment_{\lambda, \Sigma,\adversary}(0) \right] - \prob\left[ \mu \leftarrow \revokeexperiment_{\lambda, \Sigma,\adversary}(1) \right] \right| = \eps(\secparam),$$
for some non-negligible function $\eps(\cdot)$. For convenience, we model $\adversary$ as a pair of quantum algorithms $(\adversary_0,\adversary_1)$, where $\adversary_0$ corresponds to the pre-revocation adversary, and $\adversary_1$ corresponds to the post-revocation adversary.
We consider the following sequence of hybrid distributions.\\

\noindent $\hybrid_{1}^b$: This corresponds to $\revokeexperiment^{\adversary}(1^{\secparam},b)$.
\begin{enumerate}
 \item $\adversary_0$ submits two $m$-bit messages $(\mu_0,\mu_1)$ and a polynomial $k=k(\secparam)$ to the challenger.
    \item The challenger samples $y \sim \bit^m$ and produces a quantum state
       $$
\ket{S_y} = 
\frac{1}{\sqrt{2^n}}\sum_{x \in \bit^n} \ket{\varphi(x || y)}. 
$$
The challenger then sends $\ket{S_y}^{\otimes k}$ and $y \oplus \mu_b$ to $\adversary_0$. 
    \item $\adversary_0$ prepares a bipartite state on registers $\mathsf{R}$ and $\mathsf{AUX}$, and sends $\mathsf{R}$ to the challenger and $\mathsf{AUX}$ to $\adversary_1$.
    \item The challenger performs the projective measurement $\left\{ \ketbra{S_y}{S_y}^{\otimes k},\ \mathbb{I} - \ketbra{S_y}{S_y}^{\otimes k} \right\}$ on $\mathsf{R}$. If the measurement succeeds, the challenger outputs $\bot$. Otherwise, the challenger continues.
    \item The challenger grants $\adversary_1$ quantum oracle access to $\varphi^{-1}$. 
    \item $\adversary_1$ outputs a bit $b'$. 
\end{enumerate}

\ \\

\noindent $\hybrid_{2}^b$: This is the same experiment as in $\hybrid_{1}^b$, except that we change how $\ket{S_y}^{\otimes k}$ is generated before revocation:
\begin{itemize}
    \item The challenger samples a random subset $S \subseteq \bit^{n+m}$ of size $|S|=2^n$.
    \item The challenger samples a random $y \sim \bit^m$.

    \item The challenger sends $\ket{S}^{\otimes k}$
and $y \oplus \mu_b$ to $\adversary_0$.
\end{itemize}
Later, the challenger samples a random permutation $\pi: \bit^{n+m} \rightarrow \bit^{n+m}$ subject to the constraint that $\pi(s) = \ast||y$, for all $s \in S$. In other words, $\pi(S) = T_y$, where 
$$
T_y := \big\{ x \in \bit^{n+m} \, : \, x = (\ast||y) \big\}$$
and $|S| = |T_y|=2^n$. After revocation, $\adversary_1$ receives oracle access to $\pi$.
\\

\begin{claim}
 $\hybrid_{1}^b$ and $\hybrid_{2}^b$ are identically distributed.  
\end{claim}
\begin{proof}
This follows immediately. We just changed the order in which we sample things.
\end{proof}

\ \\

\noindent $\hybrid_{3}^b$: This is the same experiment as in $\hybrid_{2}^b$, except that we change the second part of the experiment: after the challenger sends $\ket{S}^{\otimes k}$
and $y \oplus \mu_b$ to $\adversary_0$, he does the following:
\begin{itemize}
    \item The challenger samples a random function $g: S \rightarrow T_y$. \footnote{Note that we have $g(S) \subseteq T_y$ in general.}

    \item The challenger samples a random permutation $\omega: (\bit^{n+m} \setminus S) \rightarrow (\bit^{n+m} \setminus T_y)$. 
\end{itemize}
After revocation, $\adversary_1$ receives oracle access to $f:\bit^{n+m} \rightarrow \bit^{n+m}$, where
$$
f(x) = \begin{cases}
g(x), \quad \text{ if } x \in S\\
\omega(x) , \quad \text{ if } x \notin S.
\end{cases}
$$
\\
\begin{claim}
 $\hybrid_{2}^b$ and $\hybrid_{3}^b$ are $O(q^3/2^n)$-close whenever $\adversary$ makes $q$ queries.
\end{claim}
\begin{proof}
Here, we apply Zhandry's result from \Cref{thm:Zhandry-result} 
 which says that random functions are indistinguishable from random permutations. Consider the following algorithm $\mathcal{B}$ which receives oracle access to $O$, which is either a random function $F: \bit^n \rightarrow \bit^n$ or a random permutation $P: \bit^n \rightarrow \bit^n$:
 \begin{enumerate}
     \item $\mathcal{B}$ samples a random subset $S \subseteq \bit^{n+m}$ of size $|S|=2^n$ and a random $y \sim \bit^m$.

    \item $\mathcal{B}$  sends $\ket{S_y}^{\otimes k}$
and $y \oplus \mu_b$ to $\adversary_0$.

\item When $\adversary_0$ replies with a bipartite state on registers $\mathsf{R}$ and $\mathsf{AUX}$, $\mathcal{B}$ performs the projective measurement $\left\{ \ketbra{S}{S}^{\otimes k},\ \mathbb{I} - \ketbra{S}{S}^{\otimes k} \right\}$ on $\mathsf{R}$. If it succeeds, $\mathcal{B}$ outputs $\bot$. Otherwise, $\mathcal{B}$ continues. 

\item $\mathcal{B}$ runs the post-revocation adversary $\adversary_1$ on input $\mathsf{AUX}$. Whenever $\adversary_1$ makes a query, $\mathcal{B}$ answers using the function $$
f(x) = \begin{cases}
(\tau \circ O\circ \sigma)(x), \quad \text{ if } x \in S\\
\omega(x) , \quad\quad\quad\quad\quad \text{ if } x \notin S
\end{cases}
$$
where we let
\begin{itemize}
    \item $\sigma$ be some canonical mapping from $S$ (with $|S|=2^n$) onto $\bit^n$, i.e., $\sigma$ is a
a function which assigns each element of $S \subseteq \bit^{n+m}$ a unique bit string in $\bit^n$.

\item $\tau$ be the function which maps each $x \in \bit^n$ to $(x||y) \in \bit^{n+m}$.
\end{itemize}
 \end{enumerate}
Note that whenever $O$ is a random permutation, the view of $\adversary$ is precisely $\hybrid_{3}^b$; whereas, if $O$ is a random function, the view of $\adversary$ is precisely $\hybrid_{4}^b$. Therefore, the claim follows from Zhandry's result in \Cref{thm:Zhandry-result}.
\end{proof}

\noindent $\hybrid_{4}^b$: This is the same experiment as in $\hybrid_{3}^b$, except that we change the second part of the experiment once again: after the challenger sends $\ket{S_y}^{\otimes k}$
and $y \oplus \mu_b$ to $\adversary_0$, he does the following:
\begin{itemize}
    \item the challenger samples a random function $f: \bit^{n+m} \rightarrow \bit^{n+m}$ subject to the constraint that $f(s) = \ast||y$, for all $s \in S$. 
\end{itemize}
After revocation, $\adversary_1$ receives oracle access to $f$.

\ \\
\begin{claim}
 $\hybrid_{3}^b$ and $\hybrid_{4}^b$ are $O\big(q^3/(2^{n+m} - 2^n)\big)$-close whenever $\adversary$ makes $q$ queries.
\end{claim}
\begin{proof}
The proof again follows from Zhandry's result in \Cref{thm:Zhandry-result}, since distinguishing  $\hybrid_{3}^b$ and $\hybrid_{4}^b$ amounts to distinguishing a random function from a random permutation mapping a random permutation $\omega: (\bit^{n+m} \setminus S) \rightarrow (\bit^{n+m} \setminus T_y)$. Since the domain and co-domain are of equal size $2^{n+m} - 2^n$, the advantage is at most $O\big(q^3/(2^{n+m} - 2^n)\big)$.   
\end{proof}

\noindent $\hybrid_{5}^b$: This is the same experiment as in $\hybrid_{4}^b$, but now we change what $\adversary_0$ receives in the pre-revocation phase:
\begin{itemize}
\item The challenger samples a random subset $S \subseteq \bit^{n+m}$ of size $|S|=2^n$.
    \item The challenger samples a random $y \sim \bit^m$.


    \item The challenger sends $(\ket{S_y}^{\otimes k},y)$ to $\adversary_0$.
\end{itemize}
After revocation, the challenger samples a random function $f: \bit^{n+m} \rightarrow \bit^{n+m}$ subject to the constraint that $f(s) = \ast||y\oplus \mu_b$, for all $s \in S$. $\adversary$ receives oracle access to $f$.\\
\ \\
\begin{claim}
 $\hybrid_{4}^b$ and $\hybrid_{5}^b$ are indentical.
\end{claim}
\begin{proof}
This follows from the fact that we have just re-labeled the variables.
\end{proof}

\noindent $\hybrid_{6}^b$: This is the same experiment as in $\hybrid_{5}^b$, except that we once again change what $\adversary_0$ receives in the pre-revocation phase:
\begin{itemize}
\item The challenger samples a random subset $S \subseteq \bit^{n+m}$ of size $|S|=2^n$.
    \item The challenger samples a random $y \sim \bit^m$.

\item The challenger samples a random $u \sim \bit^m$.

    \item The challenger sends $(\ket{S_y}^{\otimes k},y)$ to $\adversary_0$.
\end{itemize}
After revocation, the challenger samples a random function $f: \bit^{n+m} \rightarrow \bit^{n+m}$ subject to the constraint that $f(s) = \ast||u$, for all $s \in S$. $\adversary$ receives oracle access to $f$.

\ \\
Note that hybrids $\hybrid_{6}^0$ and $\hybrid_{6}^1$ are identically distributed.
By assumption, we also know that $\adversary$ distinguishes $\hybrid_{1}^0$ and  $\hybrid_{1}^1$ with non-negligible advantage. Therefore, our previous hybrid argument has shown that $\adversary$ must also distinguish $\hybrid_{5}^0$ and $\hybrid_{6}^0$ (respectively, hybrids $\hybrid_{5}^1$ and $\hybrid_{6}^1$) with advantage at least $\eps'(\secparam)$, for some non negligible function $\eps'(\secparam)$. To complete the proof, we show the following claim which yields the desired contradiction to the $k\mapsto k+1$ unforgeability property of subset states from \Cref{thm:rss:pd}.
\begin{claim}
$\mathsf{Forge}$ in Algorithm \ref{alg:forge} is a $\poly(\lambda)$-query algorithm (which internally runs $\adversary$) such that 
    \[\Pr_{\substack{S \subseteq \bit^{n+m}\\
|S|=2^n}}\Big[
(x_1,\dots,x_{k+1}) \in \mathrm{dist}(S,k+1) \, : \, (x_1,\dots,x_{k+1})\leftarrow \mathsf{Forge}^{\oracle_S}(\ket{S}^{\otimes k}) 
\Big]\geq 1/\poly(\lambda). \]
\end{claim}
\begin{proof}
Since $\hybrid_{5}^b$ and $\hybrid_{6}^b$ can be distinguished by the adversary $\adversary = (\adversary_0, \adversary_1)$ with advantage $\eps'(\secparam)$, for some non-negligible function $\eps'(\secparam)$, this implies that the post-revocation adversary $\adversary_1$ is an algorithm for which the following property holds: given as input uniformly random string $y$ and an auxiliary register (conditioned on revocation succeeding), $\adversary_1$ can distinguish
whether it is given an oracle for a function $H: \bit^{n+m} \rightarrow \bit^{n+m}$ which is random subject to the constraint that $H(s) = \ast||y \oplus \mu_b$ for all $s \in S$, or whether it is given an oracle for a function $G: \bit^{n+m} \rightarrow \bit^{n+m}$ which is a random function subject to the constraint that $G(s) = \ast||u$ for all $s \in S$. Crucially, the two functions differ precisely on inputs belonging to $S$, and are otherwise identical.

Consider the quantum extractor $\mathsf{Ext}^G(y,\mathsf{AUX})$ which is defined as follows:
  \begin{enumerate}
      \item Sample $i \leftarrow [q]$, where $q$ denotes the total number of queries made by $\adversary$.
      \item Run $\adversary_1^{G}(y,\mathsf{AUX})$ just before the $(i-1)$-st query to $G$. 
      \item Measure $\adversary_1$'s $i$-th query in the computational basis, and output the measurement outcome. 
  \end{enumerate}

From the O2H Lemma \ref{lem:O2H}, we get that
\begin{align*}
&\left|\Pr\left[\algo A_1^H(y,\mathsf{AUX}_{|\top})=1 \, : \, \substack{
S \subset \bit^{n+m}, |S|=2^n\\
y \sim \bit^m\\
(\mathsf{R},\mathsf{AUX}) \leftarrow \adversary_0(\ket{S}^{\otimes k},y)\\
H \text{ s.t. }
H(s) = \ast||y \oplus \mu_b, \forall s \in S
}\right] - \Pr\left[\algo A_1^{G}(y,\mathsf{AUX}_{|\top})=1 \, : \, \substack{
S \subset \bit^{n+m}, |S|=2^n\\
y \sim \bit^m, u \sim \bit^m\\
(\mathsf{R},\mathsf{AUX}) \leftarrow \adversary_0(\ket{S}^{\otimes k},y)\\
G \text{ s.t. }
G(s) = \ast||u, \forall s \in S
}\right] \right| \vspace{2mm}\\
&\leq 2 q \sqrt{\Pr\left[\mathsf{Ext}^G(y,\mathsf{AUX}_{|\top}) \in S \, : \, \substack{
S \subset \bit^{n+m}, |S|=2^n\\
y \sim \bit^m,  u \sim \bit^m\\
(\mathsf{R},\mathsf{AUX}) \leftarrow \adversary_0(\ket{S}^{\otimes k},y)\\
G \text{ s.t. }
G(s) = \ast||u, \forall s \in S
}\right]},
\end{align*}
where $\mathsf{AUX}_{|\top}$ corresponds to the register $\mathsf{AUX}$ conditioned on the event that the projective measurement $$\left\{ \ketbra{S}{S}^{\otimes k},\ \mathbb{I} - \ketbra{S}{S}^{\otimes k} \right\}$$ succeeds on register $\mathsf{R}$. Because the distinguishing advantage of the adversary $\adversary_1$ is non-negligible (conditioned on the event that revocation succeeds on register $\mathsf{R}$) and $q=\poly(\lambda)$, we get
\begin{align*}
\Pr\left[\mathsf{Ext}^G(y,\mathsf{AUX}_{|\top}) \in S \, : \, \substack{
S \subset \bit^{n+m}, |S|=2^n\\
y \sim \bit^m, u \sim \bit^m\\
(\mathsf{R},\mathsf{AUX}) \leftarrow \adversary_0(\ket{S}^{\otimes k},y)\\
G \text{ s.t. }
G(s) = \ast||u, \forall s \in S
}\right] \geq 1/\poly(\lambda).    
\end{align*}
To complete the proof, we now show that $\mathsf{Forge}^{\oracle_S}(\ket{S}^{\otimes k})$ in Algorithm \ref{alg:forge} is a successful algorithm against the $k\rightarrow k+1$ unforgeability of subset states.  

\begin{algorithm}
\DontPrintSemicolon
\SetAlgoLined
\label{alg:forge}
\KwIn{$\ket{S}^{\otimes k}$ and a membership oracle $\oracle_S$, where $S \subset \bit^{n+m}$ is a subset.}
    
\KwOut{$x_1, \dots x_{k+1} \in \bit^{n+m}$.}


Sample uniformly random strings $y,u \sim \bit^m$;

Run $(\mathsf{R},\mathsf{AUX}) \leftarrow \adversary_0(\ket{S}^{\otimes k},y)$;

Measure $\mathsf{R}$ in the computational basis to obtain $x_1,\dots,x_k$;

Run the quantum extractor $\mathsf{Ext}^G(y,\mathsf{AUX})$ to obtain an element $x_{k+1}$, where the oracle $G$ can be simulated via $\oracle_S$ as follows: on input $x \in \bit^{n+m}$, we let
$$
G(x) = \begin{cases}
g_1(x)||u, &\text{ if } \oracle_S(x)=1\\
g_2(x), & \text{ otherwise }
\end{cases}
$$
where $g_1: \bit^{n+m} \rightarrow \bit^n$ and $g_2: \bit^{n+m} \rightarrow \bit^{n+m}$ are uniformly random functions.

Output $(x_1,\dots,x_{k+1})$.

\caption{$\mathsf{Forge}^{\oracle_S}(\ket{S}^{\otimes k})$}
\end{algorithm}
Let $\mathsf{Revoke}(S,k,\mathsf{R})$ denote the projective measurement $\{ \ketbra{S}{S}^{\otimes k},\ \mathbb{I} - \ketbra{S}{S}^{\otimes k}\}$ of register $\mathsf{R}$.
Using the Simultaneous Distinct Extraction Lemma (Lemma \ref{lem:distinct_extract}), we get that
\begin{align*}
   &\Pr_{\substack{S \subseteq \bit^{n+m}\\
|S|=2^n}}\Big[
(x_1,\dots,x_{t+1}) \in \mathrm{dist}(S,k+1) \, : \, (x_1,\dots,x_{k+1})\leftarrow \mathsf{Forge}^{\oracle_S}(\ket{S}^{\otimes k}) 
\Big]\\
&\geq  \left(1- O\left(\frac{k^2}{2^n}\right)\right) \cdot \Pr\left[\mathsf{Revoke}(S,k,\mathsf{R})=\top
\, : \, \substack{
S \subset \bit^{n+m}, |S|=2^n\\
y \sim \bit^m\\
(\mathsf{R},\mathsf{AUX}) \leftarrow \adversary_0(\ket{S}^{\otimes k},y)}
\right]\\
&\quad\quad\quad\cdot \Pr\left[\mathsf{Ext}^G(y,\mathsf{AUX}_{|\top}) \in S \, : \, \substack{
S \subset \bit^{n+m}, |S|=2^n\\
y \sim \bit^m, u \sim \bit^m\\
(\mathsf{R},\mathsf{AUX}) \leftarrow \adversary_0(\ket{S}^{\otimes k},y)\\
G \text{ s.t. }
G(s) = \ast||u, \forall s \in S
}\right] \,
\end{align*}
which is at least inverse polynomial in $\lambda$. This proves the claim.
\end{proof}
Thus, we obtain the desired contradiction to the $k\mapsto k+1$ unforgeability property of subset states from \Cref{thm:rss:pd}. This completes the proof of multi-copy revocable encryption security.
\end{proof}

\subsection{Simultaneous Distinct Extraction Lemma}\label{sec:distinct-ext}

The following lemma allows us to analyze the probability of simultaneously extracting $k+1$ distinct subset elements in some subset $S \subseteq \bit^n$ in terms of the success probability of revocation (i.e., the projection onto $\ketbra{S}{S}^{\otimes k}$) and the success probability of extracting another subset element from the adversary's state.

\begin{lemma}[Simultaneous Distinct Extraction]
\label{lem:distinct_extract}
Let $n,k \in \N$ and
let $\rho \in \mathcal{D}({\cal H}_{X} \otimes {\cal H}_{Y})$ be an any density matrix, where ${\cal H}_{X}$ is an $n\cdot k$-qubit Hilbert space and where ${\cal H}_{Y}$ is an arbitrary Hilbert space. Let $\ket{S}$ denote a subset state, for some subset $S \subseteq \bit^n$, and let $\cal{E}: \algo L(\algo H
_Y) \rightarrow \algo L(\algo H_{X'})$ be any $\mathsf{CPTP}$ map of the form 
$$
\cal{E}_{Y \rightarrow X'}(\sigma) = \Tr_{E} \left[V_{Y \rightarrow X'E}\, \sigma V_{Y \rightarrow X'E}^\dag\right], \quad \forall \sigma \in \algo D(\algo H_{Y}),
$$
for some isometry $V_{Y \rightarrow X'E}$ and $n$-qubit Hilbert space $\algo H_{X'}$. Consider the $\mathsf{POVM}$ element
$$\Lambda = \sum_{\substack{s_1,\dots,s_{k+1} \in S\\
(s_1,\dots,s_{k+1}) \in \mathrm{dist}(S,k+1)}} \ketbra{s_1,\dots,s_k}{s_1,\dots,s_k}_{X} \otimes V_{Y \rightarrow X'E}^{\dagger} (\ketbra{s_{k+1}}{s_{k+1}}_{X'} \otimes I_E) V_{Y \rightarrow X'E}.$$
Let $\rho_X = \Tr_Y[\rho_{XY}]$ denote the reduced state. Then, it holds that
$$\Tr[\Lambda \rho] \,\geq \, \left(1- O\left(\frac{k^2}{|S|}\right)\right) \cdot \Tr[\ketbra{S}{S}^{\otimes k} \rho_X]\cdot \Tr\Big[ \ketbra{S}{S} \, \algo E_{Y \rightarrow X'}(\sigma) \Big],$$
where $\sigma = \Tr[(\ketbra{S}{S}^{\otimes k} \otimes I)\rho]^{-1} \cdot  \Tr_X[(\ketbra{S}{S}^{\otimes k} \otimes I)\rho]$ is a reduced state in system $Y$.
\end{lemma}
\begin{proof}
\noindent 
Because the order in which we apply $\Lambda$ and $(\ketbra{S}{S}^{\otimes k} \otimes I)$ does not matter, we have the inequality
\begin{align}\label{ineq:gamma}
\Tr\left[\Lambda \rho \right]
&\geq \Tr\left[ \left(\ketbra{S}{S}^{\otimes k} \otimes I \right) \Lambda \rho  \right] \nonumber\\
&=  \Tr\left[ \left(\ketbra{S}{S}^{\otimes k} \otimes I \right) \Lambda \rho \left(\ketbra{S}{S}^{\otimes k} \otimes I \right)  \right]\nonumber\\
&= \Tr\left[ \Lambda  (\ketbra{S}{S}^{\otimes k} \otimes I )  \rho \left(\ketbra{S}{S}^{\otimes k} \otimes I \right) \right]. 
\end{align}
Notice also that $(\ketbra{S}{S}^{\otimes k} \otimes I)\rho(\ketbra{S}{S}^{\otimes k} \otimes I)$ lies in the image of $(\ketbra{S}{S}^{\otimes k} \otimes I)$, and thus
\begin{align}\label{ineq:sigma}
(\ketbra{S}{S}^{\otimes k} \otimes I)\rho(\ketbra{S}{S}^{\otimes k} \otimes I) = \Tr[(\ketbra{S}{S}^{\otimes k} \otimes I)\rho] \cdot ( \ketbra{S}{S}^{\otimes k} \otimes \sigma),
\end{align}
for some $\sigma \in \algo D(\algo H_Y)$.
Using \eqref{ineq:gamma} and \eqref{ineq:sigma}, we get the lower bound,
\begin{eqnarray*}
\Tr\left[\Lambda \rho \right] &\geq&  \Tr\left[ \Lambda  (\ketbra{S}{S}^{\otimes k} \otimes I )  \rho \left(\ketbra{S}{S}^{\otimes k} \otimes I \right) \right] \\
& = &  \Tr[(\ketbra{S}{S}^{\otimes k} \otimes I)\rho] \cdot \Tr\left[\Lambda\left( \ketbra{S}{S}^{\otimes k} \otimes \sigma \right) \right].
\end{eqnarray*}
Next, we analyze the final right-most quantity separately. Specifically, we find

\begin{align*}
&\Tr\left[\Lambda\left( \ketbra{S}{S}^{\otimes k} \otimes \sigma \right) \right]\\
&= \sum_{\substack{s_1,\dots,s_{k+1} \in S:\\
(s_1,\dots,s_{k+1})\\\in \mathrm{dist}(S,k+1)}} \Tr\left[  \ketbra{s_1,\dots,s_k}{s_1,\dots,s_k}_X \otimes V_{Y \rightarrow X'E}^{\dagger} \left(  \ketbra{s_{k+1}}{s_{k+1}}_{X'} \otimes I_E \right) V_{Y \rightarrow X'E} \left( \ketbra{S}{S}^{\otimes k}\otimes \sigma \right) \right] \\
&= \sum_{s_{k+1} \in S}\Bigg( \sum_{\substack{s_1,\dots,s_{k} \in S:\\
(s_1,\dots,s_{k+1})\\\in \mathrm{dist}(S,k+1)}} |S|^{-k} \Bigg)\Tr\Big[  V_{Y \rightarrow X'E}^{\dagger} (\ketbra{s_{k+1}}{s_{k+1}}_{X'} \otimes I_E)  V_{Y \rightarrow X'E} \, \sigma\Big] \\
&= \sum_{s_{k+1} \in S}\Bigg( 1 - \sum_{\substack{s_1,\dots,s_{k} \in S:\\
(s_1,\dots,s_{k}) \notin \mathrm{dist}(S,k)\\
\text{or } \, \exists i \in [k]: \, s_i = s_{k+1}
}} |S|^{-k} \Bigg)\Tr\Big[  (\ketbra{s_{k+1}}{s_{k+1}}_{X'} \otimes I_E)  V_{Y \rightarrow X'E} \,\sigma \, V_{Y \rightarrow X'E}^{\dagger}\Big] \\
&\geq  \left(1- O\left(\frac{k^2}{|S|}\right)\right)\cdot \sum_{s_{k+1} \in S}\Tr\Big[  (\ketbra{s_{k+1}}{s_{k+1}}_{X'} \otimes I_E)  V_{Y \rightarrow X'E} \,\sigma \, V_{Y \rightarrow X'E}^{\dagger}\Big] \\
&=  \left(1- O\left(\frac{k^2}{|S|}\right)\right)\cdot \sum_{s_{k+1} \in S}\Tr\left[ \ketbra{s_{k+1}}{s_{k+1}}_{X'} \Tr_E \left[V_{Y \rightarrow X'E} \,\sigma \, V_{Y \rightarrow X'E}^{\dagger}\right] \right] \\
&= \left(1- O\left(\frac{k^2}{|S|}\right)\right)\cdot  \Tr\Big[ \ketbra{S}{S} \, \algo E_{Y \rightarrow X'}(\sigma) \Big].
\end{align*}
Putting everything together, this gives the desired inequality
$$\Tr[\Lambda \rho] \,\geq \, \left(1- O\left(\frac{k^2}{|S|}\right)\right) \cdot \Tr[\ketbra{S}{S}^{\otimes k} \rho_X]\cdot \Tr\Big[ \ketbra{S}{S} \, \algo E_{Y \rightarrow X'}(\sigma) \Big],$$
where $\sigma = \Tr[(\ketbra{S}{S}^{\otimes k} \otimes I)\rho]^{-1} \cdot  \Tr_X[(\ketbra{S}{S}^{\otimes k} \otimes I)\rho]$ is a reduced state in system $Y$.
\end{proof}

\section{Multi-Copy Revocable Programs: Definition}\label{sec:revocavle programs-def}

In this section, we study whether arbitrary functionalities can be revoked. We define and study revocable programs with multi-copy security. In this notion, there is a functionality preserving compiler that takes a program and converts it into a quantum state, which can later be certifiably revoked.

We now give a formal definition of revocable programs.

\begin{definition}[Revocable Program] 
\label{def:revprogramsyntax}
Let $\mathscr{P} = \bigcup_{\lambda \in \N} \algo P_\lambda$ be a class of efficiently computable program families $\algo P_\lambda =\{P:\algo X_\lambda \rightarrow \algo Y_\lambda\}$ with domain $\algo X_\lambda$ and range $\algo Y_\lambda$.
A revocable program compiler for the class $\mathscr{P}$ is a tuple $\Sigma= (\Compile,\Eval,\Revoke)$ consisting of the following QPT algorithms:
\begin{itemize}
\item $\Compile(1^\lambda,P)$: on input the security parameter $1^\lambda$ and a program $P \in \algo P_\lambda$ with $P: \cal{X}_\lambda\rightarrow \cal{Y}_\lambda$, output a quantum state $\ket{\Psi_{P}}$ and a (private) verification key $\vk$.

\item $\Eval(\ket{\Psi_{P}},x)$: on input a quantum state $\ket{\Psi_{P}}$  and input $x \in \cal{X}_{\lambda}$, output $P(x)$.

\item $\Revoke(\vk,\sigma)$: on input the verification key $\vk$ and a state $\sigma$, output $\top$ or $\bot$.
\end{itemize}
In addition, we require that $\Sigma$ satisfies the following two properties for all $\lambda \in \N$:
\begin{description}
    \item \textbf{Correctness of evaluation:} for all programs $P \in \algo P_\lambda$ and inputs $x \in \algo X_\lambda$, it holds that
$$\prob\left[P(x) \leftarrow \Eval(\ket{\Psi_{P}},x) \ :\ (\ket{\Psi_{P}},\vk) \leftarrow \Compile(1^\lambda,P) \right] \geq 1 - \negl(\secparam).$$

\item \textbf{Correctness of revocation:} for all programs $P \in \algo P_\lambda$, it holds that
$$\prob\left[\top \leftarrow \Revoke(\vk,(\ket{\Psi_{P}}) \ :\ (\ket{\Psi_{P}},\vk) \leftarrow \Compile(1^\lambda,P)\right] \geq 1 - \negl(\secparam).$$
\end{description}
\end{definition}

\paragraph{Multi-Copy Revocable Security.}

We use the following notion of security.

\begin{definition}[Multi-Copy Revocable Security for Programs]\label{def:revprogsec} Let $\mathscr{P} = \bigcup_{\lambda \in \N} \algo P_\lambda$ be a class of efficiently computable program families $\algo P_\lambda =\{P:\algo X_\lambda \rightarrow \algo Y_\lambda\}$ and let $\mathscr{D}_{\mathscr{P}} = \bigcup_{\lambda \in \N} \algo D_{\algo P_\lambda}$ be an ensemble of program distributions. Let $\mathscr{D}_{\algo X} = \bigcup_{\lambda \in \N}\algo{D}_{\algo X_\lambda}$ be an ensemble of challenge distribution families with $\algo{D}_{\algo X_\lambda}=\{\algo D_{\algo X_\lambda}(P)\}_{P \in \algo P_\lambda}$. Consider the following experiment between a QPT adversary $\algo A$ and a challenger.\\
\ \\
\noindent \underline{$\revokeexperiment_{\lambda,\Sigma,\adversary}^{\mathscr{D}_{\mathscr{P}},\mathscr{D}_{\algo X}}$:}
\begin{enumerate}
    \item $\adversary$ submits a polynomial $k=k(\lambda)$ to the challenger. 
    \item The challenger samples 
a program $P \sim \algo D_{\algo P_\lambda}$ with domain $\algo X_\lambda$ and range $\algo Y_\lambda$,
and runs $\Compile(1^\lambda,P)$ to generate a pair $(\ket{\Psi_{P}},\vk)$. Afterwards, the challenger sends the quantum state $\ket{\Psi_{P}}^{\otimes}$ to $\adversary$. 
    \item $\adversary$ returns a quantum state $\rho$. 
    \item The challenger performs the measurement $\left\{ \ketbra{\Psi_{P}}{\Psi_{P}}^{\otimes k},\ \mathbb{I} - \ketbra{\Psi_{P}}{\Psi_{P}}^{\otimes k} \right\}$ on the returned state $\rho$. If the measurement succeeds, the game continues; otherwise, the challenger aborts.
    \item The challenger 
    samples a challenge input $x \sim \algo D_{\algo X_\lambda}(P)$
    sends $x$ to $\adversary$. 
    \item The adversary outputs $y \in \algo Y_\lambda$.

    \item The challenger outputs $1$ if and only if $P(x)=y$.
\end{enumerate}
We say that a revocable program compiler $\Sigma= (\Compile,\Eval,\Revoke)$ has multi-copy revocable security for the ensembles $\mathscr{D}_{\mathscr{P}}$ and $\mathscr{D}_{\algo X}$, if the following holds for any QPT adversary $\algo A$:
$$ \prob\left[1 \leftarrow \revokeexperiment_{\lambda,\Sigma,\adversary}^{\mathscr{D}_{\mathscr{P}},\mathscr{D}_{\algo X}} \right]  \leq p_{\mathrm{triv}}^{\mathscr{D}_{\mathscr{P}},\mathscr{D}_{\algo X}}(\lambda)+\negl(\lambda)\, ,$$
where $p_{\mathrm{triv}}^{\mathscr{D}_{\mathscr{P}},\mathscr{D}_{\algo X}}(\lambda)=\sup_{\adversary}\{\prob\left[  P(x) \leftarrow \adversary(x)\ :\ x \leftarrow \algo D_{\algo X_\lambda}(P)\right]\}$ is the trivial guessing probability.
\end{definition}

\section{Construction of Multi-Copy Secure Revocable Programs in a Classical Oracle Model} 

In this section, we give a construction of multi-copy secure revocable programs; specifically, we work with a classical oracle model. Our construction is as follows:

\begin{construction}\label{const:revprog-scheme} Let $\lambda \in \N$ be the security parameter. Let $n,m \in \N$ be polymomial in $\lambda$. Let $\Phi = \{\Phi_\lambda\}_{\lambda \in \N}$ be an ensemble of permutations $\Phi_\lambda = \{\varphi_\kappa : \bit^{n+m} \rightarrow \bit^{n+m}\}_{\kappa \in \algo K_\lambda}$, for some set $\algo K_\lambda$.
\end{construction}
\begin{itemize}
 \item $\setup(1^{\secparam})$: sample a uniformly random key $\kappa \in \algo K_\lambda$ and let $\vk = \kappa$.
\item $\Compile(1^{\secparam},P)$: on input $1^{\lambda}$ and a program $P \in \algo P_\lambda$ with $P: \cal{X}_\lambda\rightarrow \cal{Y}_\lambda$, do the following:
\begin{itemize}
    \item Sample  $y \sim \bit^m$ and prepare the subset state given by
    $$
\ket{S_y} = 
\frac{1}{\sqrt{2^n}}\sum_{x \in \bit^n} \ket{\varphi_{\kappa}(x || y)}. 
$$
 \item Let $\ket{\Psi_{P}}=\ket{S_y}$ and, for brevity, let $S = \{\varphi_\kappa(x||y) \, : \, x \in \bit^n\}$ be the corresponding subset.
 \item Let $O=O_{P,S}$ denote an a (public) classical oracle, which is defined as follows:
 \[
O_{P,S}(x,s)=
\begin{cases}
    P(x), & \text{ if } s \in S\\
    0, &\text{ otherwise }
\end{cases}
\]
  \end{itemize}  
   \item $\Eval^{O}(\ket{\Psi_{P}},x)$: on input $\ket{\Psi_{P}}$, $x \in \cal{X}_{\lambda}$, do the following:
  \begin{itemize}
       \item Coherently evaluate $O_{P,S}$ on input $\ket{x} \otimes \ket{\Psi_{P}}$, and compute its output into an ancillary register.
       \item Measure the ancillary register and then output the measurement outcome.
   \end{itemize}
   \item $\Revoke(\vk,\rho)$: on input $\vk$, a state $\rho$ and verification key $\vk$, it parses  $y \leftarrow \vk$ and applies the measurement $\left\{ \ketbra{S_y}{S_y},\ \mathbb{I} - \ketbra{S_y}{S_y} \right\}$ to $\rho$; it outputs $\top$ if it succeeds, and $\bot$ otherwise.
    \end{itemize}
    The above scheme is easily seen to satisfy correctness.

\paragraph{Proof of Multi-Copy Revocable Security.}

Before we analyze the security of \Cref{const:revprog-scheme}, let us first remark that we can instantiate the scheme using a QPRP family $\Phi_\lambda = \{\varphi_\kappa : \bit^{n+m} \rightarrow \bit^{n+m}\}_{\kappa \in \algo K_\lambda}$, for some key space $\algo K_\lambda$. In the security proof, however, we will work with the random permutation model instead. This means that we will consider random permutations throughout the security game.

\begin{theorem}\label{thm:revoc-program}

\Cref{const:revprog-scheme} satisfies multi-copy revocable  security for any pair of distributions $\mathscr{D}_{\mathscr{P}},\mathscr{D}_{\algo X}$ in a classical oracle model, where the recipient receives an (ideal classical) oracle for the purpose of evaluation.
\end{theorem}
\begin{proof}
Because we are working in the random permutation model, we will henceforth assume that $\Sigma$ in \Cref{const:permutation-scheme} is instantiated with a perfectly random permutation $\varphi$ rather than a QPRP permutation $\varphi_\kappa$.
    Suppose not. Then there exists an adversary $\adversary$ such that,
    \[ \prob\left[1 \leftarrow \revokeexperiment_{\lambda,\Sigma,\adversary}^{\mathscr{D}_{\mathscr{P}},\mathscr{D}_{\algo X}} \right]  = \eps(\secparam)+p_{\text{triv}},
\]
 
    for some non negligible function $\eps(\secparam)$. For convenience, we model $\adversary$ as a pair of QPT algorithms $(\adversary_0,\adversary_1)$, where $\adversary_0$ corresponds to the pre-revocation adversary, and $\adversary_1$ corresponds to the post-revocation adversary.
We consider the following sequence of hybrid distributions.\\
\ \\
    \noindent $\hybrid_{1,x}$: this corresponds to $\revokeexperiment_{\lambda,\Sigma,\adversary}^{\mathscr{D}_{\mathscr{P}},\mathscr{D}_{\algo X}}$.\\
    \ \\\underline{$\revokeexperiment_{\lambda,\Sigma,\adversary}^{\mathscr{D}_{\mathscr{P}},\mathscr{D}_{\algo X}}$:}
\begin{enumerate}
    \item $\adversary$ submits a polynomial $k=k(\lambda)$ to the challenger. 
    \item The challenger samples 
a program $P \sim \algo D_{\algo P_\lambda}$ with domain $\algo X_\lambda$ and range $\algo Y_\lambda$, and a random $y \leftarrow \{0,1\}^m$, and prepares $k$ copies of the subset state over $S = \{\varphi(x||y) \, : \, x \in \bit^n\}$ given by 
\[
\ket{S}=\frac{1}{\sqrt{2^n}}\sum_{x \in \{0,1\}^n}
\ket{\varphi(x||y)}.\]
\item The adversary sends $\ket{S}^{\otimes k}$ to $\adversary$, and grants $\adversary$ access to the  oracle $O=O_{P,S}$, where
 \[
O_{P,S}(x,s)=
\begin{cases}
    P(x), & \text{ if } s \in S\\
    0, &\text{ otherwise }
\end{cases}
\]

  \item The challenger performs the projective measurement $\left\{ \ketbra{S}{S}^{\otimes k},\ \mathbb{I} - \ketbra{S}{S}^{\otimes k} \right\}$ on $\mathsf{R}$. If the measurement succeeds, the challenger outputs $\bot$. Otherwise, the challenger continues.
\item The challenger 
    samples a challenge input $x \sim \algo D_{\algo X_\lambda}(P)$
    sends $x$ to $\adversary$. 
    \item The adversary outputs $y \in \algo Y_\lambda$.
    \end{enumerate}
    \noindent $\hybrid_{2,x}:$ This is the same as hybrid 1 except, after the revocation phase, $\adversary$ has access to the oracle $\tilde{O}_{P,S}$ which always outputs 0.
\begin{claim}
   $\hybrid_{2,x}$ and $\hybrid_{1,x}$ are computationally  indistinguishable.
\end{claim}

\begin{proof}
Let's assume for contradiciton that  $\hybrid_{2,x}$ and $\hybrid_{1,x}$ can be distinguished by the adversary $\adversary = (\adversary_0, \adversary_1)$ with advantage $\eps'(\secparam)$, for some non-negligible function $\eps'(\secparam)$. This implies that the post-revocation adversary $\adversary_1$ is an algorithm for which the following property holds: given as input a string $x\leftarrow \cal{D}_{\cal{X}}(P)$, and an auxiliary register $\mathsf{AUX}_{|\top}$ (conditioned on revocation succeeding), $\adversary_1$ can distinguish
whether it is given an oracle for a function $O_{P,S}$,  which is defined as follows:
 \[
O_{P,S}(x,s)=
\begin{cases}
    P(x), & \text{ if } s \in S\\
    0, &\text{ otherwise }
\end{cases}
\]
or whether it is given an oracle for a function $\tilde{O}_{P,S}$, which always outputs 0. Crucially, the two functions differ precisely on inputs belonging to $S$, and are otherwise identical.

Consider the quantum extractor $\mathsf{Ext}^{O_{P,S}}(x,\mathsf{AUX})$ which is defined as follows:
  \begin{enumerate}
      \item Sample $i \leftarrow [q]$, where $q$ denotes the total number of queries made by $\adversary$.
      \item Run $\adversary_1^{O_{P,S}}(x,\mathsf{AUX})$ just before the $(i-1)$-st query to $O_{P,S}$. 
      \item Measure $\adversary_1$'s $i$-th query in the computational basis, and output the measurement outcome. 
  \end{enumerate}
Note that, since we are working with the random permutation model, the subset $S$ is effectively a uniformly random subset of size $2^n$.
From the O2H Lemma \ref{lem:O2H}, we get that
\begin{align*}
&\left|\Pr\left[\algo A_1^{O_{P,S}}(x,\mathsf{AUX}_{|\top})=1 \, : \, \substack{
S \subset \bit^{n+m}, |S|=2^n\\
x \sim \bit^m\\
(\mathsf{R},\mathsf{AUX}) \leftarrow \adversary_0(\ket{S}^{\otimes k})\\
O_{P,S}
}\right] - \Pr\left[\algo A_1^{\tilde{O_{P,S}}}(x,\mathsf{AUX}_{|\top})=1 \, : \, \substack{
S \subset \bit^{n+m}, |S|=2^n\\
x \sim \bit^m\\
(\mathsf{R},\mathsf{AUX}) \leftarrow \adversary_0(\ket{S}^{\otimes k})\\
\tilde{O}_{P,S}
}\right] \right| \vspace{2mm}\\
&\leq 2 q \sqrt{\Pr\left[\mathsf{Ext}^{\O_{P,S}}(\mathsf{AUX}_{|\top}) \in S \, : \, \substack{
S \subset \bit^{n+m}, |S|=2^n\\
x \sim \bit^m\\
(\mathsf{R},\mathsf{AUX}) \leftarrow \adversary_0(\ket{S}^{\otimes k})\\
O_{P,S}
}\right]},
\end{align*}

where $\mathsf{AUX}_{|\top}$ corresponds to the register $\mathsf{AUX}$ conditioned on the event that the projective measurement $$\left\{ \ketbra{S}{S}^{\otimes k},\ \mathbb{I} - \ketbra{S}{S}^{\otimes k} \right\}$$ succeeds on register $\mathsf{R}$. Because the distinguishing advantage of the adversary $\adversary_1$ is non-negligible (conditioned on the event that revocation succeeds on register $\mathsf{R}$) and $q=\poly(\lambda)$, we get
\begin{align*}
\Pr\left[\mathsf{Ext}^{O_{P,S}}(x,\mathsf{AUX}_{|\top}) \in S \, : \, \substack{
S \subset \bit^{n+m}, |S|=2^n\\
x \sim \bit^m\\
(\mathsf{R},\mathsf{AUX}) \leftarrow \adversary_0(\ket{S}^{\otimes k})\\
O_{P,S}
}\right] \geq 1/\poly(\lambda).    
\end{align*}
To complete the proof, we now show that $\mathsf{Forge}^{\oracle_S}(\ket{S}^{\otimes k})$ in Algorithm \ref{alg:forge1} is a successful algorithm against the $k\rightarrow k+1$ unforgeability of subset states. 
\begin{algorithm}
\DontPrintSemicolon
\SetAlgoLined
\label{alg:forge1}
\KwIn{$\ket{S}^{\otimes k}$ and a membership oracle $\oracle_S$, where $S \subset \bit^{n+m}$ is a subset.}
    
\KwOut{$x_1, \dots x_{k+1} \in \bit^{n+m}$.}


Run $(\mathsf{R},\mathsf{AUX}) \leftarrow \adversary_0(\ket{S}^{\otimes k})$;

Measure $\mathsf{R}$ in the computational basis to obtain $x_1,\dots,x_k$;

Sample a string $x \leftarrow \cal{D}_X(P)$;

Run the quantum extractor $\mathsf{Ext}^{O_{P,S}}(x,\mathsf{AUX})$ to obtain an element $x_{k+1}$, where the oracle $O_{P,S}$ can be simulated via $\oracle_S$ as follows: on input $(x,s) \in \bit^{2n+m}$, we let
$$
O_{P,S}(x) = \begin{cases}
P(x), &\text{ if } \oracle_S(s)=1\\
0, & \text{ otherwise }
\end{cases}
$$

Output $(x_1,\dots,x_{k+1})$.

\caption{$\mathsf{Forge}^{\oracle_S}(\ket{S}^{\otimes k})$}
\end{algorithm}
Let $\mathsf{Revoke}(S,k,\mathsf{R})$ denote the projective measurement $\{ \ketbra{S}{S}^{\otimes k},\ \mathbb{I} - \ketbra{S}{S}^{\otimes k}\}$ of register $\mathsf{R}$.
Using the Simultaneous Distinct Extraction Lemma (Lemma \ref{lem:distinct_extract}), we get that
\begin{align*}
   &\Pr_{\substack{S \subseteq \bit^{n+m}\\
|S|=2^n}}\Big[
(x_1,\dots,x_{t+1}) \in \mathrm{dist}(S,k+1) \, : \, (x_1,\dots,x_{k+1})\leftarrow \mathsf{Forge}^{\oracle_S}(\ket{S}^{\otimes k}) 
\Big]\\
&\geq  \left(1- O\left(\frac{k^2}{2^n}\right)\right) \cdot \Pr\left[\mathsf{Revoke}(S,k,\mathsf{R})=\top
\, : \, \substack{
S \subset \bit^{n+m}, |S|=2^n\\
x \sim \bit^m\\
(\mathsf{R},\mathsf{AUX}) \leftarrow \adversary_0(\ket{S}^{\otimes k})}
\right]\\
&\quad\quad\quad\cdot \Pr\left[\mathsf{Ext}^{O_{P,S}}(x,\mathsf{AUX}_{|\top}) \in S \, : \, \substack{
S \subset \bit^{n+m}, |S|=2^n\\
x \sim \bit^m\\
(\mathsf{R},\mathsf{AUX}) \leftarrow \adversary_0(\ket{S}^{\otimes k})\\
O_{P,S}
}\right] \,
\end{align*}
which is at least inverse polynomial in $\lambda$. This proves the claim.
\end{proof}
Notice that in $\hybrid_{2,x}$, the oracle $\tilde{O}_{P,S}$ always outputs 0, thus, if $\hybrid_{1,x}$ and $\hybrid_{2,x}$ are indistinguishable, $(\adversary_0^H,\adversary_1^H)$ succeeds with probability $\epsilon^*(\lambda)+p_{\text{triv}}$, for some non negligible $\epsilon$. However, this implies that $(\adversary_0^H,\adversary_1^H)$ breaks the revocable security of the program $P$.
Therefore, it suffices to argue that $\hybrid_{1,x}$ and $\hybrid_{2,x}$ are close in order to complete the proof.
\end{proof}

\section{Construction of Multi-Copy Secure Revocable Point Functions in the QROM}

\label{sec:pointfunction}

In this section, we construct multi-copy secure revocable point functions. Our main ingredient is an underlying revocable encryption scheme with a so-called \emph{wrong-key detection} mechanism. But first, we introduce some relevant tools, such as hybrid encryption schemes, which will be essential in our analysis.

\subsection{Hybrid Encryption}

We introduce a simple \emph{hybrid encryption} scheme which we will make use of in order to construct revocable multi-copy secure point functions. Our hybrid encryption scheme combines a revocable encryption scheme with the classical \emph{one-time pad} $\mathsf{OTP} = (\KeyGen,\Enc,\Dec)$, where an encryption of a plaintext $ m \in \bit^\lambda$ via a key $k \in \bit^\lambda$ is generated as $\mathsf{OTP}.\Enc(k,m) = k \oplus m$.

\begin{construction}[Hybrid encryption scheme]\label{cons:hybrid-encryption} Let $\lambda \in \N$ denote the security parameter and let the scheme $\Sigma= (\KeyGen,\Enc,\Dec)$ be a revocable encryption scheme. Let $\mathsf{OTP} = (\KeyGen,\Enc,\Dec,\mathsf{Revoke})$ be the classical one-time pad. We define the hybrid encryption scheme $\mathsf{HE}= (\KeyGen,\Enc,\Dec,\mathsf{Revoke})$ as follows:
\begin{itemize}
\item $\mathsf{HE}.\KeyGen(1^\lambda)$: This is the same procedure as $\Sigma.\KeyGen(1^\lambda)$.
\item $\mathsf{HE}.\Enc(\mathsf{sk},m)$: Given as input a key $\mathsf{sk}$ and a plaintext $m \in \bit^\lambda$, encrypt as follows:
\begin{enumerate}
    \item Sample a random string $r \sim \bit^\lambda$.

    \item Generate the pair $(\ket{\mathsf{ct}},\mathsf{vk}) \leftarrow \Sigma.\Enc(\sk,r)$.
    \item Output the ciphertext $(\ket{\mathsf{ct}},\mathsf{OTP}.\Enc(r,m))$ and (private) verification key $\mathsf{vk}$.
\end{enumerate}
\item $\mathsf{HE}.\Dec(\sk,c)$: given as input a key $\mathsf{sk}$ and ciphertext $c = (c_0,c_1)$, decrypt as follows:
\begin{enumerate}
    \item Compute
$r' = \Sigma.\Dec(\sk,c_0)$
\item Output the plaintext $\mathsf{OTP}.\Dec(r',c_1)$. 
\end{enumerate}
\item $\mathsf{HE}.\mathsf{Revoke}(\sk,\vk,\sigma)$: On input the keys $\sk$ and $\vk$, as well as the state $\sigma$, output $\Sigma.\mathsf{Revoke}(\sk,\vk,\sigma)$
\end{itemize}

\end{construction}

We now prove the following theorem. Note that in the theorem below, we make use of the search variant of multi-copy revocable security which we mentioned in \Cref{search-security}.

\begin{theorem}\label{thm:hybrid-enc}Let $\lambda \in \N$ denote the security parameter and 
let $\Sigma = (\KeyGen,\Enc,\Dec,\mathsf{Revoke})$ be a revocable multi-copy (search) secure encryption scheme. Then, the resulting hybrid encryption scheme given by $\mathsf{HE}= (\KeyGen,\Enc,\Dec,\mathsf{Revoke})$ in \Cref{cons:hybrid-encryption} which is instantiated with $\Sigma$ is also a revocable multi-copy (search) secure encryption scheme.
\end{theorem}
\begin{proof}
Suppose that $\adversary$ is a successful adversary against the (search) revocable security game with respect to $\mathsf{HE}$, i.e., there exists a non-negligible function $\eps(\lambda)$ such that
$$ \prob\left[ 1 \leftarrow \revokeexperiment_{\lambda, \mathsf{HE},\adversary} \right] =\eps(\secparam).$$
For convenience, we model $\adversary$ as a pair of quantum algorithms $(\adversary_0,\adversary_1)$, where $\adversary_0$ corresponds to the pre-revocation adversary, and $\adversary_1$ corresponds to the post-revocation adversary.

Consider the following reduction $\algo B=(\algo B_0,\algo B_1)$ which breaks the security of $\Sigma$.
\begin{enumerate}
   \item $\algo B_0$ runs $\adversary_0$ to obtain $k=k(\secparam)$, where $k(\secparam)$ is a polynomial, and sends it to the challenger.
    \item Upon receiving a ciphertext of the form $\Enc(\sk,m)^{\otimes k}$, for a random $m \sim \bit^\lambda$, $\algo B_1$ samples a random $r \sim \bit^\lambda$, and runs $(\mathsf{R},\mathsf{AUX}) \leftarrow \algo A_0(\Enc(\sk,m)^{\otimes k},r)$. 
    \item $\algo B_0$ sends $\mathsf{R}$ to the challenger, and sends $\mathsf{AUX}$ to the post-revocation adversary $\algo B_1$.
    \item Upon receiving $\sk$ from the challenger, $\algo B_1$ runs $m' \leftarrow \algo A_1(\sk,\mathsf{AUX})$ and outputs $r \oplus m'$. 
\end{enumerate}
Therefore, $\algo B$ succeeds with non-negligible probability $\eps(\lambda)$, which yields the desired contradiction.
\end{proof}

\subsection{Quantum Encryption with Wrong-Key Detection}

We use the following variant\footnote{\cite{Coladangelo2024quantumcopy} consider the more general notion of quantum secret-key encryption schemes.} of \emph{wrong-key detection} for quantum encryption schemes of classical messages (QECM) schemes~\cite{Coladangelo2024quantumcopy}.

\begin{definition}[Wrong-key detection for QECM schemes]\label{def:wkd} A QECM scheme $\Sigma  = (\KeyGen,\Enc,\Dec)$ satisfies the wrong-key detection (WKD) property if, for every $k'\neq k \leftarrow \KeyGen(1^\lambda)$ and plaintext $m$,
$$
\Tr\left[(I - \proj{\bot}) \Dec_{k'} \circ \Enc_k(m)\right] \leq \negl(\lambda).
$$

\end{definition}

Next, we give a simple transformation that converts a multi-copy-search-secure revocable encryption scheme to one with the WKD property in the quantum random oracle model (QROM).

\begin{construction}[Generic Transformation for WKD in the QROM]\label{cons:key-detection}
Let $\Sigma=(\KeyGen,\Enc,\Dec,\mathsf{Revoke})$ be a revocable encryption scheme and let $\lambda$ be the security parameter. Fix a function $H: \{0,1\}^{2\lambda} \rightarrow \{0,1\}^\ell$. Then, the scheme  $\Sigma^H=(\KeyGen^H,\Enc^H,\Dec^H,\mathsf{Revoke}^H)$ is defined by the following $QPT$ algorithms:
\begin{itemize}
    \item  $\KeyGen^H(1^{\secparam})$: on input $1^\lambda$, run $\KeyGen( 1^\lambda)$ to output a key $\sk \in \bit^\lambda$.
    \item $\Enc^H({\sk},m)$: on input $m \in \bit^\lambda$, sample $x \sim \{0,1\}^\lambda$, and output $(\Enc(\sk,x),H(\sk||x),x \oplus m)$.
    \item $\Dec^H(\sk,\mathsf{ct})$: on input $\mathsf{ct}$, first parse $(\rho,c,y) \leftarrow \mathsf{ct}$. Then, run $\Dec(\sk,\rho)$ and check (using an auxiliary register) whether $(\sk||\cdot)$ (i.e., the outcome with pre-fix $\sk$) maps to $c$ under $H$. If yes, measure the register with outcome $x'$, and output $x' \oplus y$. Otherwise, output $\proj{\bot}$.

\item $\mathsf{Revoke}^H(\sk,\vk,\sigma)$: On input the keys $\sk$ and $\vk$, as well as the state $\sigma$, output the result of running the procedure $\mathsf{Revoke}(\sk,\vk,\sigma)$ on the corresponding quantum part of the ciphertext.
\end{itemize}
\end{construction}

Clearly, correctness is preserved.
We now show the following lemma. For simplicity, we consider the standard notion of multi-copy revocable encryption security and remark that the oracular variant of security is analogous.

\begin{lemma}\label{lem:generic_transf}
Let $\Pi$ be any multi-copy (search) secure revocable encryption scheme and let $H: \{0,1\}^{2\lambda} \rightarrow \{0,1\}^\ell$ be a hash function, for $ \ell = 4\lambda$. Then, the scheme $\Sigma^H$ in Construction \ref{cons:key-detection} yields a multi-copy (search) secure revocable encryption scheme with WKD in the QROM.
\end{lemma}

\begin{proof}

Correctness is clearly preserved. Let us first verify the WKD property of the construction $\Sigma^H =(\KeyGen^H,\Enc^H,\Dec^H,\mathsf{Revoke}^H)$ in the QROM. We will analyze the probability that decryption with a wrong key $k'$ passes the key detection test.
Formally, for every $\kappa'\neq \kappa \leftarrow \KeyGen^H(1^\lambda)$ and plaintext $m \in \bit^\lambda$, we bound the probability as follows:
\begin{align*}
\Pr_{H}[\Dec_{\kappa'}^H( \Enc_\kappa^H(m)) \neq \bot] &= \Pr_{H,x}[H(\kappa'||\Dec_{\kappa'}(\Enc_\kappa(x)))=H(\kappa||x)]\\
&\leq \Pr_H[\exists w'\neq w\in \{0,1\}^{2\lambda}:H(w')=H(w)]\\
&\leq \frac{2^{2\lambda}-1}{2^{4\lambda}} = \negl(\lambda).
\end{align*}

For security, suppose that $\mathcal{A}$ is a successful adversary against the search variant of the revocable security game with respect to $\Sigma^H$, i.e., there exists a non-negligible function $\eps(\lambda)$ such that\\
$$ \prob\big[1 \leftarrow \revokeexperiment_{\lambda,\Sigma^H,\adversary} \big]  =\eps(\secparam).$$
For convenience, we model $\adversary$ as a pair of quantum algorithms $(\adversary_0^H,\adversary_1^H)$, where $\adversary_0^H$ corresponds to the pre-revocation adversary, and $\adversary_1^H$ corresponds to the post-revocation adversary. Suppose also that $\algo A^H$ makes at most $q=\poly(\lambda)$ many queries to the random oracle $H$.

Now, consider the following hybrids:\\
\ \\
\noindent $\hybrid_{1}$: This corresponds to $\revokeexperiment_{\lambda,\Sigma^H,\adversary}$:
\begin{enumerate}
   \item $\adversary_0$  submits $k=k(\secparam)$, where $k(\secparam)$ is a polynomial.  
    \item The challenger samples a message $m \sim \bit^\lambda$ and $x \sim \bit^\lambda$, lets $\sk \leftarrow \KeyGen(1^{\secparam})$ and produces $k$ copies of $\ket{\psi} \leftarrow \enc(\sk,x)$. The challenger then sends $(\ket{\psi}^{\otimes k},H(\sk||x),x \oplus m)$ to $\adversary_0^H$. 
    \item $\adversary_0^H$ prepares a bipartite state on registers $\mathsf{R}$ and $\mathsf{AUX}$, and sends $\mathsf{R}$ to the challenger and $\mathsf{AUX}$ to the post-revocation adversary $\adversary_1^H$.
    \item The challenger performs the projective measurement $\left\{ \ketbra{\psi}{\psi}^{\otimes k},\ \mathbb{I} - \ketbra{\psi}{\psi}^{\otimes k} \right\}$ on register $\mathsf{R}$. If the measurement succeeds, the game continues. Otherwise, the challenger outputs $\bot$.
    \item The challenger sends $\sk$ to $\adversary_1^H$. 
    \item $\adversary_1^H$ outputs a message $m'$. 

    \item The outcome of the experiment is $1$, if and only if $m'=m$.
\end{enumerate}

\noindent $\hybrid_{2}$: This is the same experiment as $\hybrid_{2}$, except that in Step $2$ the challenger sends $(\ket{\psi}^{\otimes k}, z, x \oplus m)$
instead of $(\ket{\psi}^{\otimes k},H(\sk||x),x \oplus m)$ to $\adversary_0$, where $z \sim \bit^\ell$ is uniformly random.\ \\
\par
Notice that in $\hybrid_{2}$, the oracle $H$ is a completely random function and $z$ is also completely random, i.e., both are independent of the rest of the scheme. Thus, if $\hybrid_{1}$ and $\hybrid_{2}$ are indistinguishable, $(\adversary_0^H,\adversary_1^H)$ breaks the revocable security of the hybrid encryption scheme in \Cref{thm:hybrid-enc}, since both $H$ and $z$ can be simulated.
Therefore, it suffices to argue that $\hybrid_{1}$ and $\hybrid_{2}$ are close in order to complete the proof.

\begin{claim}
    $\hybrid_{1}$ and $\hybrid_{2}$ are negligibly close.
\end{claim}
\begin{proof}
Suppose there exists an adversary $\adversary = (\adversary_0,\adversary_1)$ that can distinguish the two hybrids with non-negligible advantage
$\epsilon(\secparam)$. In other words, $\adversary$ satisfies
\begin{align}\label{eq:dist-adv}
&\vline\Pr\left[\algo A_1^H(\sk,\mathsf{AUX}_{|\top})=1 \, : \, \substack{
\sk \leftarrow \KeyGen(1^{\secparam})\\
m \sim \bit^\lambda, \, x \sim \bit^\lambda\\
\ket{\psi} \leftarrow \enc(\sk,x)\\
H \sim \{H:\bit^{2\lambda} \rightarrow \bit^{\ell}\}\\
(\mathsf{R},\mathsf{AUX}) \leftarrow \adversary_0^H(\ket{\psi}^{\otimes k},H(\sk||x),x \oplus m)
}\right] \nonumber\\
&\quad- \Pr\left[\algo A_1^H(\sk,\mathsf{AUX}_{|\top})=1 \, : \, \substack{
\sk \leftarrow \KeyGen^H(1^{\secparam})\\
m \sim \bit^\lambda, \, x \sim \bit^\lambda\\
\ket{\psi} \leftarrow \enc(\sk,x)\\
H \sim \{H:\bit^{2\lambda} \rightarrow \bit^{\ell}\}\\
z \sim \bit^\ell\\
(\mathsf{R},\mathsf{AUX}) \leftarrow \adversary_0^H(\ket{\psi}^{\otimes k},z,x\oplus m)
}\right]\vline \, \geq \, \eps(\lambda).
\end{align}
Our goal is to use the O2H Lemma (\Cref{lem:O2H}) in order to argue that we can use $\adversary^H = (\adversary_0^H,\adversary_1^H)$ to extract $(\sk||x)$. To this end, we introduce the following notation. For $H:\bit^{2\lambda} \rightarrow \bit^{\ell}$, a pre-image $x \in \bit^{2\lambda}$ and image $z \in \bit^\ell$, we use $H_{x,z}$ to denote the re-programmed function
$$
H_{x,z}(w) := \begin{cases}
H(x), & \text{ if } w \neq x\\
z, & \text{ if } w=x.
\end{cases}
$$
Consider the following oracle-aided quantum algorithm $\tilde{\adversary}$ which receives access to an oracle $\oracle$ and an input of the form $(\sk,x,\ket{\psi}^{\otimes k}, z, x \oplus m)$:
\begin{enumerate}
    \item $\tilde{\adversary}$ runs $(\mathsf{R},\mathsf{AUX}) \leftarrow \adversary_0^\oracle(\ket{\psi}^{\otimes k},z,x\oplus m)$.

    \item $\tilde{\adversary}$ performs the projective measurement $\left\{ \proj{\Enc(\sk,x)}^{\otimes k}, \mathbb{I} - \proj{\Enc(\sk,x)}^{\otimes k} \right\}$ on register $\mathsf{R}$. If the measurement fails, $\tilde{\adversary}$ aborts.
    \item $\tilde{\adversary}$ runs $\adversary_1^\oracle(\sk,\mathsf{AUX}_{|\top})$ and outputs whatever it outputs.
\end{enumerate}
Notice that $\tilde{\algo A}$ has precisely the same distinguishing advantage $\epsilon(\lambda)$ in \Cref{eq:dist-adv}.
Next, we also define the quantum extractor $\mathsf{Ext}^{\oracle}(\sk,x,\ket{\psi}^{\otimes k}, z, x \oplus m)$ acting as follows:
  \begin{enumerate}
      \item Sample $i \leftarrow [q]$, where $q$ denotes the total number of queries made by $\tilde{\adversary}$.
      \item Run $\tilde{\adversary}^\oracle(\sk,\ket{\psi}^{\otimes k}, z, x \oplus m)$ just before the $(i-1)$-st query to $\oracle$. 
      \item Measure $\tilde{\adversary}$'s $i$-th query in the computational basis, and output the measurement outcome. 
  \end{enumerate}

Applying the O2H Lemma (Lemma \ref{lem:O2H}) to the \emph{particular} extractor $\mathsf{Ext}^{\oracle}$ we defined above, we get 
\begin{align*}
\eps(\lambda) \,\,\leq \,\,\, \,&\vline\Pr\left[\tilde{\algo A}^H(x,\sk,\ket{\psi}^{\otimes k}, z, x \oplus m)=1 \, : \, \substack{
\sk \leftarrow \KeyGen(1^{\secparam})\\
m \sim \bit^\lambda, \, x \sim \bit^\lambda\\
\ket{\psi} \leftarrow \enc(\sk,x)\\
H \sim \{H:\bit^{2\lambda} \rightarrow \bit^{\ell}\}\\
z \sim \bit^\ell, \, H \leftarrow H_{\sk||x,z}
}\right] \\
&\quad\quad- \Pr\left[\tilde{\algo A}^G(x,\sk,\ket{\psi}^{\otimes k}, z, x \oplus m)=1 \, : \, \substack{
sk \leftarrow \KeyGen(1^{\secparam})\\
m \sim \bit^\lambda, \, x \sim \bit^\lambda\\
\ket{\psi} \leftarrow \enc(\sk,x)\\
G \sim \{G:\bit^{2\lambda} \rightarrow \bit^{\ell}\}\\
z \sim \bit^\ell
}\right] \vline \vspace{2mm}\\
&\leq 2 q \sqrt{\Pr\left[ (\sk||x) \leftarrow \mathsf{Ext}^G(\sk,x,\ket{\psi}^{\otimes k}, z, x \oplus m) \, : \, \substack{
sk \leftarrow \KeyGen(1^{\secparam})\\
m \sim \bit^\lambda, \, x \sim \bit^\lambda\\
\ket{\psi} \leftarrow \enc(\sk,x)\\
G \sim \{G:\bit^{2\lambda} \rightarrow \bit^{\ell}\}\\
z \sim \bit^\ell
}\right]}.
\end{align*}
In other words, the quantum extractor succeeds with probability at least
\begin{align*}
\Pr\left[ (\sk||x) \leftarrow \mathsf{Ext}^G(\sk,x,\ket{\psi}^{\otimes k}, z, x \oplus m) \, : \, \substack{
sk \leftarrow \KeyGen(1^{\secparam})\\
m \sim \bit^\lambda, \, x \sim \bit^\lambda\\
\ket{\psi} \leftarrow \enc(\sk,x)\\
G \sim \{G:\bit^{2\lambda} \rightarrow \bit^{\ell}\}\\
z \sim \bit^\ell
}\right] \geq \frac{\eps^2}{4q}.    
\end{align*}
But this implies that we can find $(\sk||x)$ with $1/\poly(\lambda)$ probability by measuring a random query which occurs either as a result of running the pre-revocation adversary $\algo A_0^G(\ket{\psi}^{\otimes k},z,x\oplus m)$ or of the post-revocation adversary $\algo A_1^G(\sk,\mathsf{AUX}_{|\top})$. In either case, we can use $x$ to immediately find $m$. Because $G$ is a completely random function and $z$ is also completely random (i.e., both can be simulated) this breaks the revocable security of the hybrid encryption scheme in \Cref{thm:hybrid-enc}.

\end{proof}

\subsection{Revocable Multi-Copy-Secure Point Functions}
We are now ready to define our revocable multi-copy-secure scheme for multi-bit point functions, which we obtain from any revocable encryption scheme with the aforementioned ``wrong-key detection mechanism''. Here, we consider multi-bit point functions $P_{y,m}$ of the form
$$ P_{y,m} (x) = \begin{cases} m   & \text{if } x = y\,,\\
    0^\lambda &\text{if } x \neq y \,,   \end{cases}  $$ 
where $y,m \in \{0,1\}^\lambda$.
Point functions are simply a special case of programs, so for a formal definition of syntax and revocable multi-copy security, see Definition \ref{def:revprogramsyntax}
 and Definition \ref{def:revprogsec} respectively.

 Our construction is the following

\begin{construction}[Construction for revocable multi-copy secure  point functions]\label{cons:cp_UQE}\ \\Let $\lambda\in \mathbb{N}$  be the security parameter. To construct a revocable multi-copy secure scheme for multi-bit point functions with input and output sizes $\lambda$, respectively, let $\Pi = (\KeyGen,\Enc,\Dec,\mathsf{Revoke})$ be a revocable encryption scheme with WKD, with security parameter and message length equal to $\lambda$. We define the scheme $\Sigma=(\mathsf{Compile},\mathsf{Eval},\mathsf{Revoke})$ as follows:

\begin{itemize}
\item $\mathsf{Compile}(1^{\lambda}, P_{y,m})$: on input a security parameter $\lambda$ and a multi-bit point function $P_{y,m}$, succinctly specified by the marked input $y$ (of size $\lambda$) and message $m$ (of size $\lambda$), it 

\begin{itemize}
    \item generates a pair $(\ket{\Psi_{P}},\Pi.\mathsf{vk}) \leftarrow \Pi.\Enc(y,m)$, and
    \item outputs the quantum state $\ket{\Psi_{P}}$ and (private) verification key $\mathsf{vk}=(y,\Pi.\mathsf{vk})$.
\end{itemize}

\item $\mathsf{Eval}( \ket{\Psi_{P}} , x)$:    On input $\ket{\Psi_{P}}$ and a string $x \in \{0,1\}^\lambda$ , do the following: append an ancillary qubit in the $\ket{0}$ state. Then, coherently perform a two-outcome measurement to check whether $\Pi.\Dec(x,\ket{\Psi_{P}})$ is in the state $\ket{\bot}\bra{\bot}$, or not, and store the resulting bit in the ancilla. If true, output $0^\lambda$. Otherwise, rewind the procedure and measure in the standard basis to obtain a message $m'$.

\item $\mathsf{Revoke}(\mathsf{vk},\sigma):$ on input the key $\vk$ and a state $\sigma$, parse $(y,\Pi.\mathsf{vk})\leftarrow \mathsf{vk}$ and then run $\Pi.\mathsf{Revoke}(y,\Pi.\mathsf{vk},\sigma)$.
\end{itemize}
\end{construction}

We now prove the security of Construction~\ref{cons:cp_UQE}.

\begin{theorem}\label{thm:point-function-CP-from-UCE}
Let $\Pi = (\KeyGen,\Enc,\Dec,\mathsf{Revoke})$ be any multi-copy revocable secure encryption scheme which has the WKD property. Then, Construction \ref{cons:cp_UQE} yields a revocable multi-copy secure  scheme for multi-bit point functions with respect to the pair of ensembles  $\mathscr{D}_{\mathscr{P}},\mathscr{D}_{\algo X}$ in the QROM, where $\mathscr{D}_{\mathscr{P}}$ is the (natural) uniform distribution over multi-bit point functions and $\mathscr{D}_{\algo X}$ is arbitrary.
\end{theorem}

\begin{proof} The correctness of the scheme follows directly from the WKD property of $\Pi$. Let $\mathcal{A}$ denote the adversary for $\revokeexperiment_{\lambda,\Sigma,\adversary}^{\mathscr{D}_{\mathscr{P}},\mathscr{D}_{\algo X}}$. We consider two cases, namely when $p_{\mathrm{triv}}^{\mathscr{D}_{\mathscr{P}},\mathscr{D}_{\algo X}}(\lambda)=1$ and when $p_{\mathrm{triv}}^{\mathscr{D}_{\mathscr{P}},\mathscr{D}_{\algo X}}(\lambda) <1$. In the former case, the scheme is trivially secure by definition and we are done. Hence, we will assume that $p_{\mathrm{triv}}^{\mathscr{D}_{\mathscr{P}},\mathscr{D}_{\algo X}}(\lambda)<1$ for the remainder of the proof. Note that, in this case, $\mathcal{D}_{\algo X_{\lambda}}(P_{y,m})$ has non-zero weight on the marked input $x$, for a random multi-bit point function $P_{y,m}$.

Let $x \leftarrow \mathcal{D}_{\algo X_{\lambda}}(P_{y,m})$ denote the input received by $\adversary$ during the post revocation phase. 
We can express the probability that $\mathcal{A}$ succeeds at $\revokeexperiment_{\lambda,\Sigma,\adversary}^{\mathscr{D}_{\mathscr{P}},\mathscr{D}_{\algo X}}$ as follows:
\begin{align}
&\Pr[ \mathcal{A} \text{ wins}] \nonumber\\
&\! =\! \Pr[ \mathcal{A} \text{ wins} \,|\, y \!\neq\! x] \!\cdot\!\Pr[y \!\neq\! x]\!+\!\Pr[ \mathcal{A} \text{ wins} \,|y \!=\! x] \!\cdot\!\Pr[ y \!=\! x].\\
&\! \leq\Pr[ \mathcal{A} \text{ wins} \,|\, y \!\neq\! x]\!+\!\Pr[ \mathcal{A} \text{ wins} \,|y \!=\! x] \!\cdot\!\Pr[ y \!=\! x].\\
&\! =p_{\mathrm{triv}}^{\mathscr{D}_{\mathscr{P}},\mathscr{D}_{\algo X}}(\lambda)+\!\Pr[ \mathcal{A} \text{ wins} \,|y \!=\! x] \!\cdot\!\Pr[ y \!=\! x].
\end{align}
where the second equality follows because when $x\neq y$, then the post-revocation adversary $\cal{A}_1$ does not have access to the secret key $y$, and the winning probability is simply the optimal guessing probability in this case. 

We complete the proof by showing that $\Pr[ \mathcal{A} \text{ wins} \,|\,  y = x] \leq \negl(\lambda)$. This would imply that
\begin{align}
\Pr[ \mathcal{A} \text{ wins}] \leq p_{\mathrm{triv}}^{\mathscr{D}_{\mathscr{P}},\mathscr{D}_{\algo X}}(\lambda) + \negl(\lambda).\label{eq:CP_security_bound}
\end{align}
Suppose that $\mathcal{A}$ succeeds with non negligible probability on the challenge $x$ such that $x =y$. We will use $\mathcal{A}$ to construct an adversary against the multi-copy (search) revocable security of the scheme $\Pi$. For convenience, we model $\adversary$ as a pair of QPT algorithms $(\adversary_0,\adversary_1)$, where $\adversary_0$ corresponds to the pre-revocation adversary, and $\adversary_1$ corresponds to the post-revocation adversary. Consider the $QPT$ adversary $\mathcal{B}= (\mathcal{B}_0,\mathcal{B}_1)$ against $\Pi$, which we define as follows:
\begin{enumerate}
 \item $\algo B_0$ runs $\adversary_0$ to obtain $k=k(\secparam)$, where $k(\secparam)$ is a polynomial, and sends it to the challenger.
    \item Upon receiving a multi-copy ciphertext of the form $\ket{\psi}^{\otimes k}$, where $\ket{\psi}=\Enc(y,m)$, for a random $m \sim \bit^\lambda$, $\algo B_0$ sends $\ket{\psi}^{\otimes k}$ to $\adversary_0$, and runs $(\mathsf{R},\mathsf{AUX}) \leftarrow \algo A_0(\ket{\psi}^{\otimes k})$. 
    
    \item $\algo B_0$ sends $\mathsf{R}$ to the challenger. 
    \item The challenger performs the projective measurement $\left\{ \ketbra{\psi}{\psi}^{\otimes k},\ \mathbb{I} - \ketbra{\psi}{\psi}^{\otimes k} \right\}$ on $R$. If the measurement outcome is 1 then the challenger outputs $\bot$. Otherwise, the challenger continues.
    
    \item $\algo B_0$ sends $\mathsf{AUX}$ to the post-revocation adversary $\algo B_1$.

    \item $\algo B_1$  receives the secret key $\sk=y$ from the challenger, and it forwards $y$ to  $\adversary_1$.  
    \item Finally, $\algo B_1$ outputs the output obtained by running $\adversary_1$.
\end{enumerate}
$\Pr[B \text{ wins }]= \Pr[A \text{ wins }| x=y]=\text{non-negl}(\lambda)$, which is a contradiction.
\end{proof}

Finally, we end this section with the following corollary, which is immediate from the previous results in this section.

\begin{corollary}
There exist revocable multi-bit point functions which satisfy (an oracular notion of) multi-copy security in the quantum random oracle model.     
\end{corollary}

\end{proof}

\section{Applications to Sponge Hashing}\label{sec:sponge}

In this section, we show that the techniques we developed in this paper, i.e., our analysis of $k \mapsto k+1$ unforgeability of random subsets via query lower bounds for permutations, are more broadly applicable and extend to other cryptographic settings as well. Here, we single out the so-called \emph{sponge construction} and show how to prove new space-time trade-offs for finding elements in sponge hash tables. Our key insight is that the aforementioned query problem can be viewed as a particular instance of $k \mapsto k+1$ unforgeability of random subsets in the presence of membership oracles.

\subsection{Single-Round Sponge Hashing}

The National Institute of Standards and Technology (NIST) recently announced a new international hash function standard known as SHA-3. This is a new family of cryptographic functions based on the idea of \emph{sponge hashing}~\cite{KeccakSponge3,KeccakSub3}.
This particular approach  allows for both variable input length and variable output length, which makes it particularly attractive towards the design of cryptographic hash functions. 
The internal state of a sponge function gets updated through successive applications of a so-called \emph{block function} (modeled as a random permutation) $\varphi: \bit^{r+c} \rightarrow \bit^{r+c}$, where we call the parameters $r \in \N$ the \emph{rate} and $c \in \N$ the \emph{capacity} of the sponge. 

The basic (single-round) sponge hash function $\mathsf{Sp}_{\mathsf{IV}}^\varphi: \bit^{r} \rightarrow \bit^r$ on input $x \in \bit^r$ is defined by $h = \mathsf{Sp}_{\mathsf{IV}}^\varphi(x)$, where $h$ corresponds to the first $r$ bits of $\varphi(x||\mathsf{IV}) := (h||w)$, for some fixed \emph{initialization vector} $\mathsf{IV} \in \bit^c$.
Here, the post-fix $w \in \bit^c$ is simply discarded during the evaluation of the hash function, and we can think of $\mathsf{IV} \in \bit^c$ as a random \emph{salt}.  We give an illustration of the sponge hash function below.

\begin{figure}[H]
\begin{center}
{\small
\begin{tikzpicture}
  \draw (5,-0.75) rectangle (6,0.75) node [pos=.5]{$\varphi$}; 

\draw[-] (4.6,0.35) node[left]{$x$} -- (5,0.35);
\draw[-] (4.6,-0.35) node[left]{$\mathsf{IV}$} --(5,-0.35);
  \draw[-] (6.0,0.35) node[right]{\hspace{4mm}$\mathsf{Sp}^\varphi(x)$} --(6.4,0.35);
\draw[-] (6.0,-0.35) node[right]{\hspace{4mm}$w$} --(6.4,-0.35);
 \end{tikzpicture}
\label{fig:single-sponge}
}
\end{center}
\caption{The single-round sponge hash function.}
\end{figure}
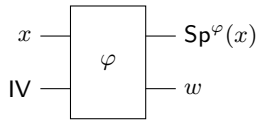

While much is known about the security of the sponge in a classical world~\cite{10.1007/978-3-540-78967-3_11}, the post-quantum security of the sponge construction in the case when the block function is a random (but invertible) permutation remains a major open problem~\cite{Czajkowski2017,Czajkowski2019,Zhandry21,carolan2024oneway,majenz2024permutationsuperpositionoraclesquantum}. The inherent difficulty in proving the post-quantum security of the sponge construction lies in the fact that an adversary needs to be modeled as a quantum query algorithm which has access to $\varphi$ and its inverse $\varphi^{-1}$.
This is where existing techniques~\cite{Czajkowski2017,Czajkowski2019} seem to break down and require an entirely new approach, e.g., as the recent works of~\cite{carolan2024oneway,majenz2024permutationsuperpositionoraclesquantum} suggest.

\subsection{Space-Time Trade-Offs for Finding Elements in Sponge Hash Tables}

We now consider a simple query problem which seems to lie outside of the scope of existing techniques in~\cite{Czajkowski2017,Czajkowski2019,carolan2024oneway,majenz2024permutationsuperpositionoraclesquantum}. Specifically, we consider the following task of finding elements in hash tables, which is also a potentially relevant problem in practice.\\ 
\ \\
\textbf{A simple query problem:} Suppose that $\varphi: \bit^{r+c} \rightarrow \bit^{r+c}$ is a random permutation, for some $r,c \in \N$. Let  $\mathsf{IV} \in \bit^c$ be a random salt. Given as input a hash table of size $S$ which consists of
$$
\big(h_1 =  \mathsf{Sp}_{\mathsf{IV}}^\varphi(x_1), \, \dots, \, h_S =  \mathsf{Sp}_{\mathsf{IV}}^\varphi(x_S)\big)
$$
where $x_1,\dots,x_S \leftarrow \bit^r$ are randomly chosen distinct inputs, how many queries to $\varphi^{-1}$ and a checking oracle $\mathsf{Valid}_{\varphi,\mathsf{IV}}$ (which outputs $1$, if an input hash is in the range of $\mathsf{Sp}_{\mathsf{IV}}^\varphi$, and $0$ otherwise) does a quantum algorithm need to find a valid new element in the range of $\mathsf{Sp}_{\mathsf{IV}}^\varphi$? Suppose that
$$
\Pr\left[ \mathsf{Valid}_{\varphi,\mathsf{IV}}(h')=1 \,\, \wedge\,\, h' \neq h_i, \forall i \in [S] \, \,: \,\,h' \leftarrow \mathcal{A}^{\varphi^{-1},\, \mathsf{Valid}_{\varphi,\mathsf{IV}}}(h_1,\dots,h_S) \right]= \epsilon > 0.
$$
Can we prove that $\epsilon$ must satisfy a space-time trade-off, i.e., in terms of the size $S$ of the hash table, and the number of queries $T$? In other words, can we show that the only way an adversary can add a new element to the hash table is by making a large number of queries? Clearly, there is a subtle trade-off that appears in this scenario, since a sufficiently large hash table could potentially make it harder to find valid new elements.

Analyzing such a query problem is non-trivial for a number of reasons. The first challenge lies in the fact the oracle-aided algorithm $\algo A$ receives as input list of hashes which depend on the salt itself (i.e., the algorithm can act \emph{adaptively}). Typically, in traditional function inversion tasks~\cite{cryptoeprint:2019/1093,CiC-1-1-27}, the adversary only receives advice on the hash function itself, i.e., it does not get to receive information which depends on the actual salt as part of the pre-processing phase. The second challenge lies in the fact that $\algo A$ receives not just one, but two oracles $\varphi^{-1}$ and $\mathsf{Valid}_{\varphi,\mathsf{IV}}$ which depend on $\varphi$, and also the salt $\mathsf{IV}$.

Surprisingly, our techniques for proving the unforgeability of random subsets via query lower bounds for permutations are amendable in this context and allow us to show the following theorem.

\begin{theorem}
Let $r,c \in \N$ be integers and let $\lambda = \min(r,c)$. Let $\varphi: \bit^{r+c} \rightarrow \bit^{r+c}$ be a uniformly random permutation and $\mathsf{IV} \leftarrow \bit^c$ be a random salt. Suppose that $\mathcal{A}^{\varphi^{-1},\, \mathsf{Valid}_{\varphi,\mathsf{IV}}}$ is a $T$-query quantum algorithm which receives as input a hash table
$\big(h_1 =  \mathsf{Sp}_{\mathsf{IV}}^\varphi(x_1), \, \dots, \, h_S =  \mathsf{Sp}_{\mathsf{IV}}^\varphi(x_S)\big)$ of size $S$,
where $x_1,\dots,x_S \leftarrow \bit^r$ are randomly chosen distinct inputs. Suppose that $\mathcal{A}$ succeeds at finding a new element in the hash table with probability
$$
\Pr\left[ \mathsf{Valid}_{\varphi,\mathsf{IV}}(h')=1 \,\, \wedge\,\, h' \neq h_i, \forall i \in [S] \, \,: \,\,h' \leftarrow \mathcal{A}^{\varphi^{-1},\, \mathsf{Valid}_{\varphi,\mathsf{IV}}}(h_1,\dots,h_S) \right]= \epsilon > 0.
$$    
Then, $\mathcal{A}$ satisfies the space-time trade-off given by
$$
\epsilon \leq O\left( T \cdot \sqrt{\frac{2^r -S}{2^{r+c}}}  + \sqrt{\frac{2^r -S}{2^{r}}} +\frac{T^3}{2^{r+c} - 2^{r}} +\frac{T^3}{2^{r}}\right).
$$
\end{theorem}

\begin{proof}
We give a direct reduction from the $S \mapsto S+1$ unforgeability experiment in \Cref{{thm:classical-unforgeability}}. Suppose we are given as input $(a_1,\dots,a_S)$ as well as a membership oracle $\oracle_A$ for a random subset $A \subseteq \bit^{r+c}$ of size $|A|=2^r$. Consider the reduction $\mathcal{B}^{\oracle_A}(a_1,\dots,a_S)$ which does the following:
\begin{enumerate}
    \item $\mathcal{B}$ samples a random salt  $\mathsf{IV} \leftarrow \bit^c$.

    \item $\mathcal{B}$ samples a random function $f: \bit^{r+c} \rightarrow \bit^{r+c}$ subject to the constraint that $f(a) = \ast||\mathsf{IV}$, whenever $a \in A$.

    \item $\mathcal{B}$ runs $h' \leftarrow \mathcal{A}^{f,\, \mathcal{V}_{f,\mathsf{IV}}}(h_1,\dots,h_S)$ to obtain a string $h' \in \bit^r$, where we let $h_1,\dots,h_S$ denote the first $r$ bits of the strings $a_1,\dots,a_S$, respectively, and where we let
    $$
\mathcal{V}_{f,\mathsf{IV}}(x) = 
\begin{cases}
1, & \text{ if }  f(x) = \ast||\mathsf{IV}\\
0, & \text{ otherwise.}
\end{cases}
$$

\item $\mathcal{B}$ does a standard Grover search to find a string $g \in \bit^c$ such that $a'_{S+1} \in A$ and $a'_{S+1} \neq a_i$, for all $i \in [k]$, with  $a'_{S+1}:=h'||g$. If no such string $g$ is found, $\mathcal{B}$ outputs $\bot$. Note that this step requires at most $O(\sqrt{2^c})$ many additional queries in the worst case.
\end{enumerate}

We now analyze the probability that $\mathcal{B}$ solves the  $S \mapsto S+1$ unforgeability problem by finding $S+1$ distinct elements in $A$. Using an identical analysis as in the proof of \Cref{thm:revoc-encryption}, we can show that the probability that $\mathcal{A}^{f,\, \mathcal{V}_{f,\mathsf{IV}}}(h_1,\dots,h_S)$ finds a valid element such that $\mathcal{V}_{f,\mathsf{IV}}(h')=1$ drops by at most
$$
O\left(\frac{T^3}{2^{r+c} - 2^{r}}+\frac{T^3}{2^{r}}\right)
$$
compared to $\epsilon$.
Moreover, $\mathcal{B}$ solves the $S \mapsto S+1$ unforgeability problem with precisely the same loss in probability.
We can now invoke \Cref{thm:classical-unforgeability} to argue that 
\begin{align*}
\epsilon - O\left(\frac{T^3}{2^{r+c} - 2^{r}}\right) -  O\left(\frac{T^3}{2^{r}}\right) \leq O\left( \left(T + \sqrt{2^c}\right) \cdot \sqrt{\frac{2^r-S}{2^{r+c}}} + \frac{2^r-S}{2^{r+c}}\right).
\end{align*}
Re-arranging for $\epsilon$ and simplifying the bound, we get that
$$
\epsilon \leq O\left( T \cdot \sqrt{\frac{2^r -S}{2^{r+c}}}  + \sqrt{\frac{2^r -S}{2^{r}}} +\frac{T^3}{2^{r+c} - 2^{r}} +\frac{T^3}{2^{r}}\right).
$$

\paragraph{Open questions.} We now state a number of interesting open questions which are relevant to this section. While our analysis of finding elements in sponge hash tables does result in an interesting space-time trade-off, we do not believe that our bounds are potentially not tight. We conjecture that, if $r$ is significantly larger than $c$, then $\mathcal{A}$ requires at least $$\Omega\left(\sqrt{\frac{2^{r+c}}{2^r-S}} \right)$$ 
many queries to succeed with constant success probability. We leave this as an interesting open problem. Finally, we believe that our analysis could also be extended in the case when $\algo A$ receives access to the permutation $\varphi$---in addition to the oracles $\varphi^{-1}$ and $\mathsf{Valid}_{\varphi,\mathsf{IV}}$. Because our techniques do not seem apply in this setting, this suggest that an entirely new approach is necessary.

\end{proof}



\printbibliography

\appendix

\section{Appendix}

\paragraph{Tighter bound in the classical case.}

Here, we give a slightly tighter variant of \Cref{thm:rss:pd}, which we apply in the context of sponge hashing.

\begin{theorem}[Classical $k \mapsto k+1$ Unforgeability]\label{thm:classical-unforgeability}
Let $n \in \N$ and $k \in \N$. Then, for any $q$-query quantum oracle algorithm $\mathcal{A}$, and any $1 \leq k < s\leq 2^n$, it holds that
\begin{align*}
&\Pr_{\substack{x_1,\dots,x_k \leftarrow S\\
(x_1,\dots,x_k) \in \mathrm{dist}(S,k)
}}\Pr_{\substack{S \subseteq \bit^n\\
|S|=s}}\Big[
x' \in S \, \wedge \, x' \neq x_i, \, \forall i \in [k] \,\, : \,\, x'\leftarrow \mathcal{A}^{\oracle_S}(x_1,\dots,x_k) 
\Big]\\
&\leq  O\left(q \cdot \sqrt{\frac{s-k}{2^n}} + \frac{s-k}{2^n}\right).
\end{align*}
In particular, for $k=\poly(n)$, $q=\poly(n)$ and $s(n)= n^{\omega(1)}$, the probability is at most $\negl(n)$.

\end{theorem}
\begin{proof}
Suppose that $S \subseteq \bit^n$ is a random subset of size $|S|=s$. We can model the quantum oracle algorithm $\mathcal{A}^{\oracle_S}$ on input $(x_1,\dots,x_k)$ as a sequence of oracle queries and unitary computations followed by a measurement. Thus, the final output state just before the measurement can be written as
$$
\ket{\Psi_q^S} = U_q \oracle_S U_{q-1} \dots U_1 \oracle_S U_0\ket{\psi_0}\ket{x_1,\dots,x_k} \,, 
$$
where $U_0,U_1,\dots,U_q$ are unitaries (possibly acting on additional workspace registers, which we omit above), and where $\ket{\psi_0}$ is some fixed initial state which is independent of $S$. 

Now let $X = \{x_1,\dots,x_k\} \subset S$. We now consider the state 
$$
\ket{\Psi_q^X} = U_q \oracle_X U_{q-1} \dots U_1 \oracle_X U_0\ket{\psi_0}\ket{x_1,\dots,x_k}.
$$
We now claim that the states $\ket{\Psi_q^S}$ and $\ket{\Psi_q^X}$ are sufficiently close.
From the definition of $\oracle_X$ and $\oracle_S$, we have that $\oracle_X(x)\neq \oracle_S(x)$ iff $x \in S\backslash X \subset \{0,1\}^{n}$. By the O2H Lemma (\Cref{lem:O2H}), we get
\begin{align*}
\E_{\substack{S \subseteq \bit^n, \,|S|=s\\
x_1,\dots,x_k \leftarrow S\\
(x_1,\dots,x_k) \in \mathrm{dist}(S,k)}}\left\| \ket{\Psi_q^S} - \ket{\Psi_q^X} \right\| 
&\leq 2q\E_{\substack{S \subseteq \bit^n, \,|S|=s\\
x_1,\dots,x_k \leftarrow S\\
(x_1,\dots,x_k) \in \mathrm{dist}(S,k)}}\sqrt{
\frac{1}{q} \sum_{i=0}^{q-1}\big\|\Pi_{S\backslash X} \ket{\Psi_{i}^X}\big\|^2}\\
&\leq 2q\sqrt{
\frac{1}{q}\sum_{i=0}^{q-1} \,\E_{\substack{S \subseteq \bit^n, \,|S|=s\\
x_1,\dots,x_k \leftarrow S\\
(x_1,\dots,x_k) \in \mathrm{dist}(S,k)}}\big\|\Pi_{S\backslash X} \ket{\Psi_{i}^X}\big\|^2} & \text{(Jensen's inequality)}\\
&= O\left(q \cdot \sqrt{\frac{s-k}{2^n}}\right).
\end{align*}
Therefore, the probability (over the choice of $S$ and $X$) that $\mathcal{A}^{\oracle_X}(x_1,\dots,x_k)$ succeeds is 
 at most the probability that $\mathcal{A}^{\oracle_S}(x_1,\dots,x_k)$ succeeds---up to an additive loss of
$O(q \cdot \sqrt{\frac{s-k}{2^n}})$. Thus,
\begin{align*}
&\Pr_{\substack{x_1,\dots,x_k \leftarrow S\\
(x_1,\dots,x_k) \in \mathrm{dist}(S,k)
}}\Pr_{\substack{S \subseteq \bit^n\\
|S|=s}}\Big[
x' \in S \, \wedge \, x' \neq x_i, \, \forall i \in [k] \,\, : \,\, x'\leftarrow \mathcal{A}^{\oracle_S}(x_1,\dots,x_k) 
\Big]\\
&\leq \Pr_{\substack{x_1,\dots,x_k \leftarrow S\\
(x_1,\dots,x_k) \in \mathrm{dist}(S,k)
}}\Pr_{\substack{S \subseteq \bit^n\\
|S|=s}}\Big[
x' \in S \, \wedge \, x' \neq x_i, \, \forall i \in [k] \,\, : \,\, x'\leftarrow \mathcal{A}^{\oracle_X}(x_1,\dots,x_k) 
\Big] + O\left(q \cdot \sqrt{\frac{s-k}{2^n}}\right)\\
&\leq O\left(q \cdot \sqrt{\frac{s-k}{2^n}} + \frac{s-k}{2^n}\right).
\end{align*}
where the last equality follows from the fact that
$$
\Pr_{\substack{x_1,\dots,x_k \leftarrow S\\
(x_1,\dots,x_k) \in \mathrm{dist}(S,k)
}}\Pr_{\substack{S \subseteq \bit^n\\
|S|=s}}\Big[
x' \in S \, \wedge \, x' \neq x_i, \, \forall i \in [k] \,\, : \,\, x'\leftarrow \mathcal{A}^{\oracle_X}(x_1,\dots,x_k) 
\Big] \leq \frac{s-k}{2^n}.
$$
This proves the claim.
\end{proof}

\end{document}